\documentclass[10pt,journal,compsoc]{IEEEtran}

\usepackage{amssymb, amsmath, amsfonts, amsthm}
\usepackage{graphicx, cite, mathtools}
\usepackage{dsfont,multicol}

\newtheorem{theorem}{Theorem}
\newtheorem{lemma}{Lemma}
\newtheorem{proposition}{Proposition}
\newtheorem{corollary}{Corollary}

\newtheorem{remark}{Remark}
\usepackage{epstopdf}
\usepackage{url}
\usepackage{amsmath}
\def\mathclap#1{\text{\hbox to 0pt{\hss$\mathsurround=0pt#1$\hss}}}
\newcommand{\diff}{\mathop{}\!d}
\DeclareMathOperator*{\argmin}{arg\,min}

\newcommand{\ceil}[1]{\left\lceil #1 \right\rceil}

\allowdisplaybreaks
\usepackage{lipsum,graphicx}
\usepackage{caption, subcaption, cite}

\IEEEoverridecommandlockouts

\begin{document}

\title{FD-JCAS Techniques for mmWave HetNets:\\ Ginibre Point Process Modeling and Analysis}

\author{Christodoulos Skouroumounis, \IEEEmembership{Member, IEEE},
	Constantinos Psomas, \IEEEmembership{Senior Member, IEEE},
	and~Ioannis Krikidis, \IEEEmembership{Fellow, IEEE}
	\thanks{Christodoulos Skouroumounis, Constantinos Psomas, and Ioannis Krikidis are with the  IRIDA Research Centre for Communication Technologies, Department of Electrical and Computer Engineering, University of Cyprus, Cyprus, e-mail:\{cskour03, psomas, krikidis\}@ucy.ac.cy.}}

\IEEEtitleabstractindextext{
\begin{abstract}
In this paper, we study the co-design of full-duplex (FD) radio with joint communication and radar sensing (JCAS) techniques in millimeter-wave (mmWave) heterogeneous networks (HetNets). Spectral co-existence of radar and communication systems causes mutual interference between the two systems, compromising both the data exchange and sensing capabilities. Focusing on the detection performance, we propose a cooperative detection technique, which exploits the sensing information from multiple base stations (BSs), aiming at enhancing the probability of successfully detecting an object. Three combining rules are considered, namely the \textit{OR}, the \textit{Majority} and the \textit{AND} rule. In real-world network scenarios, the locations of the BSs are spatially correlated, exhibiting a repulsive behavior. Therefore, we model the spatial distribution of the BSs as a $\beta$-Ginibre point process ($\beta$-GPP), which can characterize the repulsion among the BSs. By using stochastic geometry tools, analytical expressions for the detection performance of $\beta$-GPP-based FD-JCAS systems are expressed for each of the considered combining rule. Furthermore, by considering temporal interference correlation, we evaluate the probability of successfully detecting an object over two different time slots. Our results demonstrate that our proposed technique can significantly improve the detection performance when compared to the conventional non-cooperative technique.
\end{abstract}

}
\maketitle
\IEEEdisplaynontitleabstractindextext
\IEEEpeerreviewmaketitle

\IEEEraisesectionheading{\section{Introduction}\label{sec:introduction}}

\IEEEPARstart{E}{merging} applications such as smart cars, unmanned aerial vehicles (UAV) and enhanced localization, lead to an ever-increasing demand for systems with both communication and radar sensing capabilities \cite{HAN,GAU}. In order to address this demand, joint communication and radar sensing (JCAS) techniques have been developed, integrating the two operations of communication and sensing over a shared spectrum. As an emerging research topic, the spectrum sharing of JCAS techniques enables the efficient use of the available spectrum, and provides a new way for designing and modeling novel network architectures and protocols that can benefit from the synergy of communication and radar systems.

As a straightforward approach for achieving the spectral co-existence of communication and radar sensing systems, the authors in \cite{SAR,HES} consider an opportunistic spectrum sharing scheme between a pulsed radar and a cellular network, where the communication system transmits if and only if its transmission will not compromise the operation of the radar system. While such approach achieves low implementation complexity, the simultaneous operation of communication and radar sensing systems is unattainable. The main challenge for the joint operation of communication and radar sensing systems over a shared spectrum is the negative effects of the mutual interference between the two systems, which significantly alleviates the network performance \cite{ZHE}. Therefore, the employment of interference mitigation techniques is of paramount importance in the context of JCAS systems. The impact of mutual interference on the coverage performance of communication systems is well investigated in the literature, however, its impact on the sensing capabilities of radar systems has not fully elucidated. From this standpoint, the authors in \cite{RAY} study the spectral co-existence of a power-controlled cellular network with rotating radar devices, demonstrating that the power control is of critical importance for the mitigation of the mutual interference. Moreover, the spectrum sharing between a multi-input-multi-output (MIMO) radar and a wireless communication system is analyzed in \cite{DEN1}; the results show that by exploiting the degrees of freedom offered by the MIMO technology, the mutual interference can be mitigated from the main/side lobe without compromising the radar system's performance. Similarly, beamforming techniques are proposed in \cite{LIU} for a MIMO JCAS system, where each remote radio unit can operate as a radar and as a communication transmitter by simultaneously sensing nearby objects and communicating with downlink (DL) users, respectively. However, the concept of cooperation between multiple base stations (BSs) in the context of FD-JCAS systems, aiming to mitigate the overall interference and enhance the network’s performance is missing from the current literature.

Efficient spectrum utilization is another key issue for JCAS systems that needs to be addressed. Spectrum efficiency in wireless communication systems can be significantly improved by operating in a full-duplex (FD) mode, which is being considered for the next-generation wireless systems. Specifically, the FD technology could potentially double the spectral efficiency with respect to the half-duplex counterpart, as it employs simultaneous transmission and reception using non-orthogonal channels  \cite{MAHM}. However, the non-orthogonal operation creates a self-interference (SI) link between the input and the output antennas. Owing to the overwhelming negative effect of the SI at a transceiver, the FD technology has been previously regarded as an unrealistic approach in wireless communications. Fortunately, with the recent advancements in transceiver design and signal processing techniques, the SI signal can be successfully suppressed below the noise floor, and therefore the FD technology becomes feasible \cite{RII,ALE}. However, the performance of the FD radio in large-scale multi-cell networks is additionally compromised by the increase of the intra- and out-of-cell co-channel interference. Several research efforts have been carried out to study the effect of multi-user interference on the FD performance for large-scale wireless networks, and several techniques have been proposed to mitigate the additional interference caused by the FD operation \cite{RII,SAK,SKO}.

The requirements of next generation cellular networks in massive connectivity and high throughput, motivate the operation of millimeter-wave (mmWave) frequency bands and their heterogeneous network (HetNet) deployment. mmWave communications are considered as a suitable environment for integrating FD-JCAS systems due to their unique features, such as the large available bandwidth and the antenna directivity, which can boost the quality of the direct link \cite{KUM}. In particular, the large available bandwidth of mmWave communications can lead to multi-Gbps rates, which is essential to satisfy the high-capacity requirements of emerging applications in FD-JCAS systems, such as fully-connected vehicles, UAV and robotics \cite{DAN}. Furthermore, recent studies have shown that the higher path-losses of the mmWave signals and their sensitivity to blockages, can improve the network performance by mitigating the overall interference \cite{SKO,TUR}. Therefore, the modeling and the analysis of FD-JCAS systems in the context of mmWave cellular networks, is of critical importance in order to support the massive data-rate demands of emerging applications and combat the severe multi-user interference. Even though FD-JCAS systems are well-investigated for sub-$6$ GHz applications, few works deal with their operation in higher frequency bands. Recent studies suggest exploiting the radar function to support vehicle-to-everything communications based on the IEEE $802.11$ad wireless local area network protocol, which operates in the $60$ GHz band \cite{KUM,DAN}. Similarly, the authors in \cite{HOU} investigate the implementation of radar devices that operate in the mmWave frequency bands in the context of vehicular systems, and study the impact of radar interference on the probability of successful range estimation. In addition, the authors in \cite{PAR} characterize the blockage detection probability achieved by the mmWave radar devices, and provide guidelines for their efficient deployment. However, the above studies only focus on the co-design of FD-JCAS in the mmWave cellular networks, neglecting the negative effects of the SI link on the overall network performance.

Motivated by the mathematical tractability of the Poisson point process (PPP)-based abstraction model, modeling and analysis of JCAS systems with the aid of stochastic geometry has gained a lot of interest \cite{REN,MUN,LYU,WEI}. However, the PPP is only suitable for the networks where nodes are deployed in a fully unplanned fashion. Moreover, in \cite{AND2}, the authors shown that by modeling the spatial distribution of the network’s nodes as a PPP, the achieved performance of the actual deployment of cellular networks is underestimated. In many practical networks, the locations of the nodes are determined to alleviate the interference or extend the coverage region, and therefore there exists a form of repulsion among the network's nodes. In this context, the Ginibre point processes (GPPs) \cite{BLA}, which is a special case of the determinantal point processes (DPPs) \cite{LI}, have gained much attention as models for wireless networks where there exists a form of repulsion among the nodes in the networks. The $\alpha$-GPP ($-1\leq\alpha\leq 0$) \cite{LU} is a superposition of $-1/\alpha$ independent GPPs with an intensity scaled by $\sqrt{-\alpha}$. By modeling the BSs' locations of a primary network as an $\alpha$-GPP, the authors in \cite{LU} investigate the performance of the ambient backscatter communications in the context of wireless-powered HetNets, where the secondary transmitters utilize the traffic resources offered by the primary network. On the other hand, the $\beta$-GPP ($0<\beta\leq 1$) \cite{KON2} is a thinned and rescaled GPP which is obtained by deleting each point of the GPP independently with probability  $1-\beta$. In \cite{KON2}, the authors analyzed the secure communication performance by using artificial noise from the source and cooperative jamming for single antenna nodes, where the network's nodes are spatially distributed as $\beta$-GPP. In addition to the spatial correlation between the network's nodes, the presence of common randomness in the locations of the interfering nodes induces temporal correlation in the observed interference  \cite{KRI}. Although there has been substantial work quantifying interference correlation in wireless network \cite{KRI,NIG}, the investigation of temporal correlation in FD-JCAS systems has been disregarded.

To the best of our knowledge, the spectrum co-existence of FD-JCAS systems in mmWave HetNets, where the network’s nodes experience a repulsion between each other, is overlooked from the literature. In addition, the effect of the SI and the temporal interference correlation on the achieved performance of the FD-JCAS deployments has not been investigated. Hence, the aim of this work is to fill this gap by modeling and analyzing such networks and by providing new analytical results for the network performance in a stochastic geometry framework. Specifically, the main contributions of this paper are summarized as follows:
\begin{itemize}
	\item We develop a mathematical framework based on stochastic geometry, which comprises the modeling of spectrum sharing FD-JCAS systems in the context of heterogeneous mmWave cellular networks. Specifically, we consider a JCAS system, where all BSs exhibit radar sensing capabilities, of which a fraction also exhibit DL communication capabilities by exploiting the concept of FD radio. As real network deployments form a more regular point pattern than the PPP, the developed framework takes into account the correlation among the locations of the BSs by modeling their  distribution as a $\beta$-GPP. 
	\item A novel cooperative multi-point radar detection (CoMRD) technique is proposed, aiming at providing an enhanced detection probability of objects sensing by the BSs' radar system. Three hard-decision combining rules, namely the \textit{OR}, the \textit{Majority} and the \textit{AND} rules, are analyzed. Using stochastic geometry tools, analytical expressions for the detection performance are derived for each of the considered hard-decision combining rules. Moreover, we derive simplified analytical expressions for the scenario where the BSs are independently deployed within the network area, i.e. $\beta\rightarrow 0$. Finally, the temporal correlation of the interference is investigated and analytical expressions of the joint detection probability over different time slots are derived.
	\item The behavior of the detection performance achieved by our proposed technique for the considered combining rules is analyzed for different system parameters. We demonstrate that our cooperative technique can significantly improve the detection performance in the context of the considered network deployments, when compared to the conventional non-cooperative radar detection technique. In addition, our numerical results illustrate that a repulsive deployment of BSs is beneficial for the detection performance, since by increasing the distance of an interfering link results in a reduction of the caused interference. Furthermore, we show that the temporal interference correlation significantly impacts the network's detection performance.
\end{itemize}

The rest of the paper is organized as follows: Section \ref{Section2} describes the considered system model together with the channel, blockage, antenna, power allocation, detection and communication models, as well as the main definitions and properties of the $\beta$-GPP. The proposed CoMRD technique is explicitly described in Section \ref{Section3}. Section \ref{Section4} provides the main results for the detection performance for each hard-decision combining rule in the context of mmWave HetNets, where the network's nodes are distributed based on a $\beta$-GPP. In Section \ref{Section5}, we consider the effect of the interference correlation on the detection performance, providing analytical expressions for the joint detection probability over different time slots. Numerical results are presented in  Section \ref{Section6}, followed by our conclusions in Section \ref{Section7}.

\begin{table*}[t]\centering
	\caption{Summary of Notations}\label{Table1}
	\scalebox{0.85}{\begin{tabular}{| l | l || l | l |}\hline
			\textbf{Notation} & \textbf{Description}&\textbf{Notation} & \textbf{Description}\\\hline
			$K$ & Total number of network tiers & $\mu,\sigma^2_{\rm SI}$ & Nagakami-$\mu$ fading parameters \\\hline
			$\Phi_k,\lambda_k(r)$ & $\beta$-GPP of BSs in $k$-th tier of density $\lambda_k(r)$ & $p_{\rm L}(r)$ & LoS probability function\\\hline
			$\widetilde{\Phi}_k,\widetilde{\lambda}_k$ & $\beta$-GPP of interfering BSs in $k$-th tier of density $\widetilde{\lambda}_k$ & $p_{\rm L},{R}$ & parameters of LoS probability function\\\hline
			$\beta_k$ & repulsion parameter of BSs in the $k$-th tier &$\phi, G$ & Main-lobe beamwidth and gain\\\hline
			$\chi_k$ & fraction of FD-JCAS BSs in $k$-th tier &$P_k$ & Maximum transmit power of BSs  \\\hline
			$\mathcal{C}$ & set of $K$ cooperative BSs  &${P}_k(d)$ & Allocated BSs' transmit power for the DL communication\\\hline
			$d$ & distance between a BS and its associated user &$\rho,\varepsilon$ & Power control parameters\\\hline
			$h(\tau)$ & power of the channel fading at time slot $\tau$  &$A$ & Cross-section area of the point-of-interest \\\hline
			$L(X,Y)$ & large-scale path-loss between nodes $X$ and $Y$ &$\sigma_n^2$ & variance of the thermal noise \\\hline
			$a,\epsilon$ & large-scale path-loss parameters  &$\delta_k$ & Binary decision of the cooperative BS from the $k$-th tier\\\hline
			$g_{\rm SI}$ & channel power gain of the residual SI &	$\mathds{1}_A$ & Indicator function that the event $A$ holds\\\hline
	\end{tabular}}
\end{table*}

\textit{Notation:} $r_x$ denotes the Euclidean norm of $x \in \mathbb{R}^d$ i.e., $r_x=||x||$, $B(x,r)$ represents a circle of radius $r$ centered at $x\in\mathbb{R}^2$, $\binom{n}{k} = \frac{n!}{k!(n-k)!}$ is the binomial coefficient, $\mathcal{G}(a, b)$ is a gamma random variable (r.v.) with shape parameter $a$ and scale parameter $b$; $\Gamma[a]$ and $\Gamma[a,x]$ denote the complete and incomplete gamma function, respectively, $\mathds{1}_X$ is the indicator function where $\mathds{1}_X=1$ if $X$ is true, otherwise $\mathds{1}_X=0$, and ${}_2F_1[\cdot,\cdot,\cdot,\cdot]$ denotes the hypergeometric function.

\section{System model}\label{Section2}
In this section, we provide details of the considered system model. The network is studied from a large-scale point-of-view leveraging tools from stochastic geometry. In order to model and analyze the spectral coexisting communication and radar systems in the context of mmWave HetNets, an FD-JCAS system is investigated, where the mmWave BSs have two operations i.e., DL transmission and object detection. The main mathematical notation related to the system model is given in Table \ref{Table1}.

\subsection{Network topology}
\begin{figure}
	\centering\includegraphics[width=0.97\linewidth]{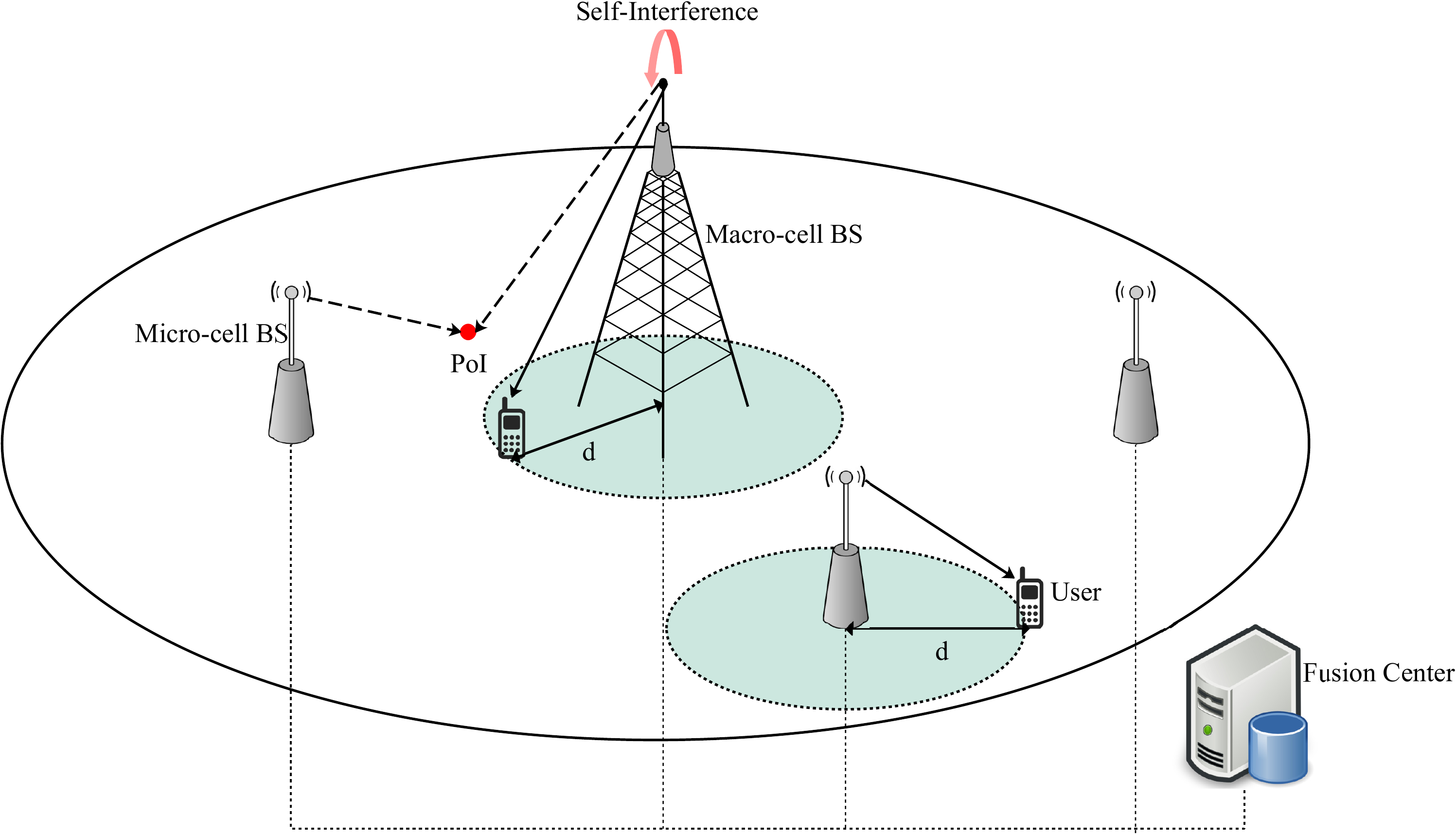}
	\caption{Network topology of a two-tier FD-JCAS system in mmWave cellular networks.}\label{fig:sfig2}
\end{figure}
We consider a HetNet composed by $K$ network tiers, consisting of randomly located BSs that operate in the mmWave frequency band. In this context, we model the spatial distribution of the BSs that belong to the $k$-th tier i.e., $k\in\{1,\dots,K\}$, according to a $\beta$-GPP denoted as $\Phi_k\!=\!\{X_{j,k}\!\in\mathbb{R}^2, j\in\mathbb{N}\}$ with spatial density $\lambda_{k}$ and repulsion parameter $\beta_k$ i.e., $\Phi_k\sim{\rm GPP}(\lambda_{k},\beta_k)$. The locations of the points-of-interest (PoIs) at which it is necessary to detect whether or not an object exists, follow a PPP $\Phi_{\rm PoI}$ with density $\lambda_{PoI}$, where at each PoI there is an object with a probability $\zeta$, where $\zeta\in[0,1]$.

The knowledge of such information is essential in many communication techniques, such as beamforming, handover and localization techniques, especially for systems that operate in high frequency bands i.e., mmWave and Tera-Hz communication \cite{KUM,DAN,HOU},  where the detection of an object is crucial for their operations. Although the acquisition of this knowledge can be accomplished by a plethora of systems, the demand for a new infrastructure prohibitively increases the implementation cost, making such techniques unsuitable for the next-generation systems \cite{HOU}. Instead of having two separate systems, we develop a JCAS technique that integrates the two functions into one by sharing hardware and signal processing modules. With such techniques, the implementation cost is significantly reduced, and the network performance can be further enhanced from the mutual sharing of information \cite{KUM,DAN,HOU}.

We assume that all BSs exhibit radar sensing capabilities and are responsible for determining whether or not there is an object in a PoI. In this work, a round-robin scheduling mechanism is employed where the detection of all PoIs within the coverage area of a BS is randomly scheduled at different time slots. In addition to the sensing capabilities, a fraction of BSs is also exhibits communication capabilities, serving their associated users in the DL direction by implementing the FD-JCAS scheme. Specifically, a BS that belongs to th $k$-th tier, exhibits both communication and radar sensing capabilities based on a predefined probability $\chi_k$, otherwise, it solely exhibits sensing capabilities. We assume that the two-dimensional space $\mathbb{R}^2$ is divided into several large sub-regions $V\subset\mathbb{R}^2$, where all BSs of each sub-region $V$ are connected to the dedicated central unit, also known as fusion center (FC), of that particular sub-region through an ideal report channel \cite{MUN}. The focus of this work is the performance investigation of BSs that belong in a single sub-region $V\in\mathbb{R}^2$, while the cooperation between different sub-regions is an interesting future work. Fig. \ref{fig:sfig2} shows a realization of a two-tier mmWave cellular network i.e, $K=2$, consisting of macro-cell and micro-cell BSs, where solid and dashed lines represent the communication signals and the radar sensing of a BS, respectively. Based on Slivnyak's theorem and without loss of generality, we focus on detecting the existence or absence of an object at the typical PoI, which is assumed to be located at the origin.

All wireless signals are assumed to experience both large-scale path-loss effects and small-scale fading. Since the small-scale fading in mmWave links is less severe than the conventional systems due to deployment of directional antennas, all links are assumed to be subject to independent Nakagami fading. Hence, the power of the channel fading is a Gamma r.v. with shape parameter $\nu$ and scale parameter $1/\nu$. For the large-scale path-loss between nodes $X$ and $Y$, we assume a bounded path-loss model which is based on the distance $d=\|X-Y\|$, in the form $L(d) = 1/\left(\epsilon+d^{a}\right)$, where $a > 2$ denotes the path-loss exponent and $\epsilon>0$. The spectrum co-existence of radar and communication systems using non-orthogonal channels causes the existence of SI between the output and the input antennas of the communication and radar systems, respectively. Regarding the channel between radar and communication antennas, we assume that the BSs employ imperfect cancellation mechanisms \cite{RII,SKO}. As such, we consider that the residual SI channel coefficient follows a Nagakami distribution with parameters $\left(\mu,\sigma^2_{\rm SI}\right)$, where $1/\sigma^2_{\rm SI}$ characterizes the SI cancellation capability of the BSs. Therefore, the power gain of the residual SI channel follows a Gamma distribution \cite{SKO} with mean $\mu$ and variance $\sigma^2_{\rm SI}/\mu$ i.e., $g_{\rm SI}\sim\mathcal{G}\left(\mu,\frac{\sigma^2_{\rm SI}}{\mu}\right)$.

\subsection{Blockage and sectorized antenna model}
A mmWave link can be either line-of-sight (LoS) or non-LoS, depending on whether the transmitter is visible to the receiver or not, due to the existence of blockages. In particular, a transmitter is considered LoS by a receiver, if and only if their communication link of length $r$ is unobstructed by blockages. We define the LoS probability function $p_{\rm L}(r)$ as the probability that a link of length $r$ is LoS. To simplify the mathematical derivation of the analysis, we consider the generalized \textit{LoS ball} model \cite{AND}, where the LoS probability function can be approximated by a step function. Specifically, a link of length $r$ is in LoS state with probability $p_{\rm L}(r)= p_{\rm L}$ if $r\leq R$, otherwise $p_{\rm L}(r)= 0$, where $R$ is the maximum length of an LoS link \cite{AND} and the constant $p_{\rm L}\in[0,1]$ is the average fraction of the LoS area in the ball of radius $R$. In this paper, the interference effect from the non-LoS signals is ignored, as we assume the dominant interference is caused by the LoS signals \cite{AND}. 

For modeling the antenna directionality of the BSs, we adopt an ideal cone antenna pattern \cite{MUN}. The antenna array gain is parameterized by two values: $1)$ the main-lobe beamwidth $\phi\in[0,2\pi]$, and $2)$ the main-lobe gain $G$ (dB). We assume that the main-lobe of the radar's antenna for all BSs is facing towards the direction of their associated PoI and that perfect beam alignment is achieved between each user and its serving BS, using the estimated angles of arrival. On the other hand, the beams of interfering links are assumed to be randomly oriented with respect to each other. Therefore, an interfering BS causes interference to another node only if the alignment of their two randomly oriented antenna patterns overlap and their link is LoS; by simple geometrical arguments, this event has probability, $\left(\frac{\phi}{2\pi}\right)^2p_{\rm L}$.

\subsection{Power allocation}
We assume that all BSs that belong to the $k$-th tier, where $k\in\{1,\dots,K \}$, are allocated with the same transmit power $P_k$, where $P_k>P_i$ for $k<i$. Due to the overwhelming negative effect of the SI at the radar's receive antenna, the DL power control for the communication system is of paramount importance. Hence, we assume that all BSs which implement the FD-JCAS scheme, utilize distance-proportional fractional power control for the DL communication with their associated user, which is located over a circle of radius $d$ centered at the BS, where $d\leq R$. The power control scheme aims at compensating the large-scale path-loss and maintaining the average received signal power at their associated user equal to $\rho$ \cite{SKO}. To accomplish this, a BS at distance $d$ from its associated user, adapts its transmitted signal power to $\rho L(d)^\varepsilon$, where $0\leq\varepsilon\leq 1$ is the power control fraction. Based on the general power control mechanism, the transmission power allocated to a BS that belongs in the $k$-th tier, can be expressed as $P_k(d)=\min\{\rho L(d)^\varepsilon,\ P_k\}$, where the BSs which are unable to fully invert the path-loss, transmit with maximum power $P_{k}$. It is important to mention here that, for the case where $\varepsilon=1$, the path-loss is completely compensated if $P_{k}$ is sufficiently large, and if $\varepsilon=0$, no path-loss inversion is performed and all BSs from the same tier transmit with the same power. The considered distance-proportional DL power control technique only requires the knowledge of the locations of the users, opposed to the truncated channel-inversion power control scheme, that demands further computational resources for the channel estimation. Several approaches have been proposed for the acquisition of this knowledge, such as low-rate feedback schemes \cite{SKO} or pilot-based sub-$6$ GHz networks \cite{KRIS}. For the detection system, we consider a fixed power transmission allocation scheme i.e., BSs from the $k$-th tier transmit with power $P_k$. It is important to mention that, for the sake of simplicity, both systems of an FD-JCAS BS have equal transmit power.

\subsection{$\beta$-GPP preliminaries}\label{bGPP}
In this subsection, we provide some background on the $\beta$-GPP; readers are referred to \cite{DEC,DEN1,BLA} for further details. The $\beta$-GPP represents a repulsive point process where the parameter $\beta\in(0,1]$ can be used to ``interpolate'' smoothly from the GPP to the PPP. More specifically, when $\beta$ becomes larger, the points in the $\beta$-GPP experience strong repulsion, and therefore the points tend to be more spread out. On the other hand, when $\beta\rightarrow 0$, the correlation among the points disappears, which corresponds to the case where the locations of the points are independent i.e., the PPP case.

Let $\Phi = \{X_j\}_{j\in\mathbb{N}}$, be a $\beta$-GPP with density $\lambda$ and repulsion parameter $\beta$ i.e., $\Phi\sim{\rm GPP}(\lambda,\beta)$. Then, the distribution of the set $\{\|X_j\|^2\}_{j\in\mathbb{N}}$ can be characterized by the distribution of the set $\{\widetilde{B}_i\}_{i\in\mathbb{N}}$, which is obtained by independently retaining from $\{B_i\}_{i\in\mathbb{N}}$ each $B_i$ with probability $\beta$, where $\{B_i\}_{i\in\mathbb{N}}$ are independent and Gamma distributed r.v. with a shape parameter $i$ and a scale parameter $\beta/t$ i.e., $B_i\sim\mathcal{G}(i,\beta/t)$, where $t$ is the scaling parameter used to control the intensity and is equal to $t=\pi\lambda$ \cite{DEN1}. Therefore, the term $\|x\|^{-a}$ can be represented as $\widetilde{B}_i^{-\frac{a}{2}} = B_i^{-\frac{a}{2}}\Xi_{i}$, where $\{\Xi_{i}\}_{i\in\mathbb{N}}$ represents the set of discrete r.v., such that $\mathbb{P}\left(\Xi_{i}=1\right)=\beta$ and $\mathbb{P}\left(\Xi_{i}=0\right)=1-\beta$. Based on the aforementioned, the following proposition provides the probability density function (pdf) of $B_{i}$ \cite{BLA}.
\begin{proposition}\label{FirstLemma}
	The pdf of a r.v. $B_{i}\sim\mathcal{G}(i,\beta/t)$, is given by
	\begin{equation}\label{CPpdf}
	f_{B_{i}}(y) = y^{i-1}\exp\left(-\frac{t}{\beta}y\right)\Bigg/\left(\left(\frac{\beta}{t}\right)^i\Gamma[i]\right),
	\end{equation}
	where $t > 0$ and $i\in\mathbb{N}$.
\end{proposition}
\setcounter{equation}{2}
\begin{figure*}[t]
	\begin{equation}
	\mathcal{P}_c^k\! \leq\! p_{\rm L}\!\sum_{\xi=1}^\nu(-1)^{\xi+1}\binom{\nu}{\xi}\exp\!\left(-s\sigma_n^2\right)\!\prod_{\substack{j=1\\ j\neq i}}^\infty\!\left(1\!-\!\beta_k\!+\!\int\nolimits_{\sqrt{d}}^{\sqrt{{R}}}\frac{\beta_k f_{C_{j,k}}(y)}{1\!+\!\frac{sP_k(d)G^2}{\left(\epsilon+y^\frac{a}{2}\right)}}dy\right)\!\prod_{\substack{z=1\\ z\neq k}}^K\prod_{j=1}^\infty\!\left(1\!-\!\beta_z\!+\!\int\nolimits_{0}^{\sqrt{{R}}}\frac{\beta_z f_{C_{j,z}}(y)}{1+\frac{sP_z(d)G^2}{\left(\epsilon+y^\frac{a}{2}\right)}}dy\right),\label{EquationC}
	\end{equation}
	\hrulefill
\end{figure*}
\setcounter{equation}{1}
In the context of the considered system model, we denote as $\{\Xi_{j,k} \}_{j\in\mathbb{N}}$, the set of discrete r.v. for the $k$-th tier, such that $\mathbb{P}\left(\Xi_{j,k}=1\right)=\beta_k$ and $\mathbb{P}\left(\Xi_{j,k}=0\right)=1-\beta_k$. Also, for $\Phi_k\sim {\rm GPP}(\lambda_{k},\beta_k)$, we define a set of independent gamma r.v. $\{B_{j,k}\}_{j\in\mathbb{N}}$, where $B_{j,k}\sim\mathcal{G}\left(j,\beta_k/(\pi\lambda_{k})\right)$  \cite{DEN}. Let $\widetilde{\Phi}_k$ depicts the locations of the active interfering BSs. Based on the considered network deployment, the point process $\widetilde{\Phi}_k$ is a thinned version of the original $\beta$-GPP $\Phi_k$, and have a density equal to $\widetilde{\lambda}_k(r)=\lambda_{k}\left(\frac{\phi}{2\pi}\right)^2p_{\rm L}\chi_k$. Therefore, for $\widetilde{\Phi}_k\sim{\rm GPP}(\widetilde{\lambda}_k(r),\beta_k)$, we define a set of independent gamma r.v. $\{C_{j,k}\}_{j\in\mathbb{N}}$, where $C_{j,k}\!\sim\!\mathcal{G}(j,\beta_k/\!(\pi\widetilde{\lambda}_k(r)))$ \cite{DEN}.

\subsection{Detection model}\label{RadSection}
\begin{figure}
	\centering\includegraphics[width=0.85\linewidth]{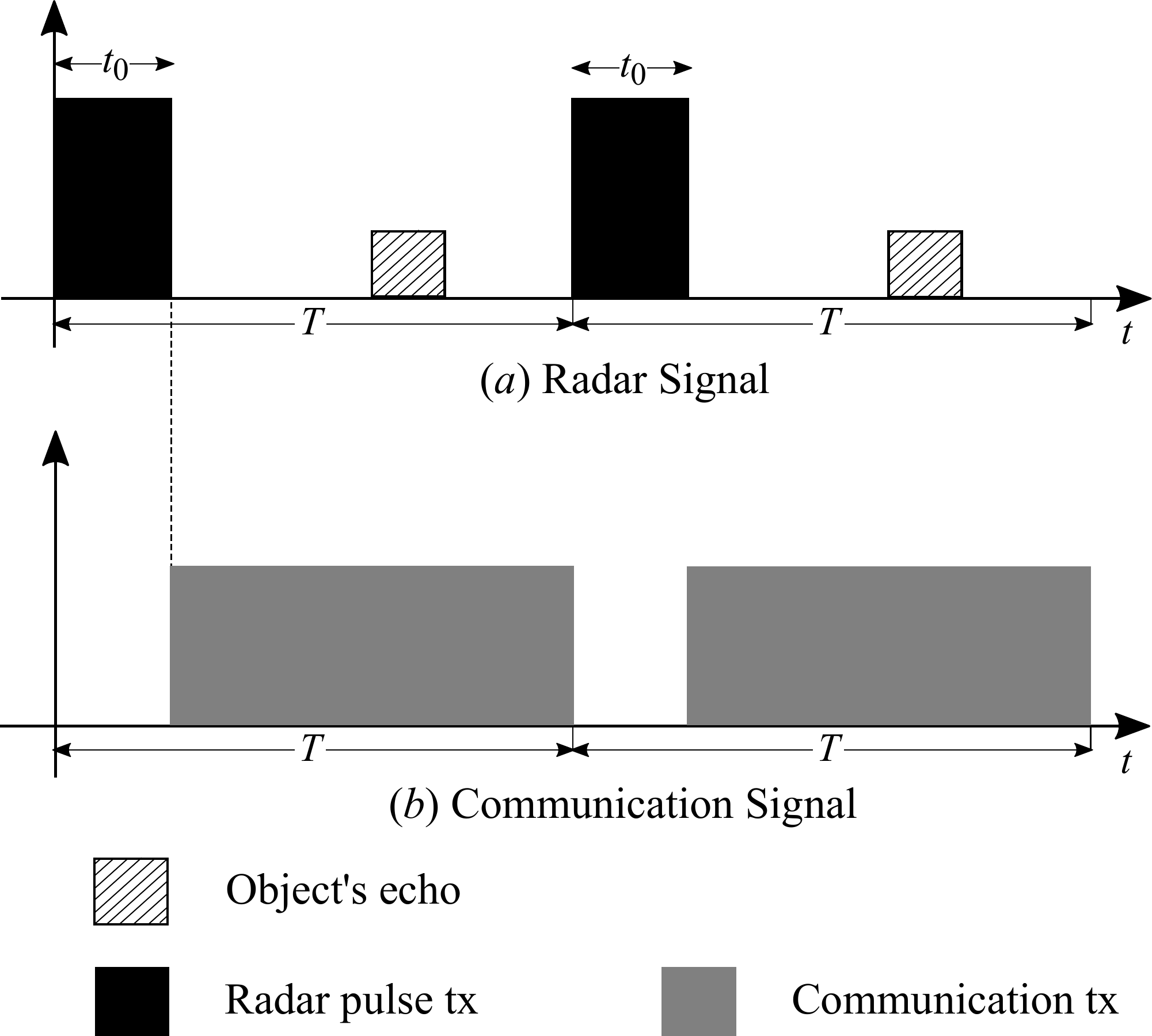}
	\caption{An example for the timeline of the operation of both radar and communication system of a transmitter that employs the considered FD-JCAS scheme. }\label{TimeSlot}
\end{figure}
An example timeline for the considered network deployment is depicted in Fig. \ref{TimeSlot}, highlighting the radar and communication operations of the FD-JCAS BSs. More specifically, the time is slotted, and a BS follows a regular pattern of duration $T$ slots (or time units). Regarding the radar operation (Fig. \ref{TimeSlot} $(a)$), within the first $t_0$ time units, where $t_0\ll T$, we assume that the BS at $x\in\Phi_k$ firstly broadcasts a narrow pilot sequence towards its main-lobe direction. During the remaining $T-t_0$ time units, the BS measures the reflected signal power received within its main-lobe direction, due to the existence of the object at $y\in\mathbb{R}^2$. The echo measured during the time slot $\tau$ at the BS, follows from the well-known radar equation as \cite{HOU,PAR,REN}
\begin{equation*}
P^{(\tau)}_{\rm refl}(r)= \!\frac{P_kGh_0(\tau)L(r)}{4\pi}\frac{AL(r)}{4\pi}A_e=\!P_kG^2\frac{A\ell}{4\pi}h_{0}(\tau)L^2(r),
\end{equation*}
where $r=\|x-y\|$, $G$ is the transmit antenna gain, $h_0(\tau)$ denotes the channel power gain of the link between the BS and the object during time slot $\tau$, $A$ is the radar cross-section area of the object, $A_e$ is the effective area i.e., $A_e=G\frac{c^2}{4\pi f^2}$, where $c$ is the speed of light and $f$ is the carrier frequency, and $\ell = (c/4\pi f)^2$. Then, by focusing on a particular time slot, the signal-to-interference-plus-noise ratio (SINR) of the reflected pilot sequence, is given by
\begin{equation}\label{SINR}
\gamma_{k}^{(\tau)}(r)= \frac{P^{(\tau)}_{\rm refl}(r)}{\mathcal{I}+I_{\rm SI}\left(d,k\right)\mathds{1}_{ \text{JCAS}}+\sigma_n^2},
\end{equation}
where $\mathcal{I}=\sum\nolimits_{z=1}^K\sum\nolimits_{X_{j,z}\in\widetilde{\Phi}_z}P_z(d)G^2\ell {h}_{j,z}(\tau)L(\|X_{j,z}-x\|)$ is the total interfering power observed by the BS at $x\in\Phi_k$ caused by all active BSs from the $i$-th tier i.e., $X_{j,z}\in\widetilde{\Phi}_z$, $j\in\mathbb{N}$, $h_{j,z}(\tau)\sim\Gamma[\nu,1/\nu]$ denotes the channel power gain at time slot $\tau$ between the active interfering BS at $X_{j,z}\in\widetilde{\Phi}_z$ and the BS at $x\in\Phi_k$; $\mathds{1}_\text{JCAS}$ represents the indicator function of the event ``the BS at $x\in\Phi_k$ implements the FD-JCAS scheme'' and  $\sigma_n^2$ denotes the variance of the thermal noise. Due to the non-orthogonal operation of communication and radar systems (see Fig. $2$), the residual SI observed at the receive antenna of the radar system after the SI cancellation is equal to $I_{\rm SI}(d,k) = P_k\left(d\right)g_{\rm SI}=\min\{\rho L(d)^\varepsilon,\ P_k\}g_{\rm SI}$, where $g_{\rm SI}$ is the power gain of the residual SI channel.

For the successful detection of an object we assume two conditions. Firstly, the SINR of the reflected pilot signal should be larger than a specific threshold. Note that this is a typical condition in the context of detection problems \cite{HOU}. Hence, by assuming that the corresponding BS is at a distance $v$ from the typical PoI, the probability of this first condition is evaluated as $\mathbb{P}[\gamma_k^{(\tau)}(v)>\theta]$. Second, there should be no blockage between the typical PoI and its associated BS, as well as an object must exists at the typical PoI. If the typical PoI is blocked by a blockage, the BS cannot detect the typical PoI successfully. The probability of this second condition is $\zeta p_{\rm L}$. Considering both conditions, the successful detection probability of a BS that belongs in the $k$-th tier at the time slot $\tau$, is given by $\zeta p_{\rm L}\mathbb{P}[\gamma_k^{(\tau)}(v)>\vartheta]$. Conversely, for the case where a blockage exists between the object at the typical PoI and the BS, and the signal-to-interference ratio of the signal reflected at the blockage is larger than the detection threshold, a false alarm (FA) is triggered. Hence, the FA probability of a BS that belongs in the $k$-th tier and detect a blockage at distance $u$, is given by $(1-p_{\rm L})\zeta\mathbb{P}[P_{\rm refl}(u)/\mathcal{I}>\vartheta]$ \cite{HOU}. Based on the above-mentioned performance metrics and the Neyman-Pearson Lemma \cite{REN}, the minimum threshold, $\theta$, for the decision rule that ensures a tolerable pre-defined false alarm rate is assessed in Section \ref{Section4}.

\subsection{Communication model}\label{CommSection}
Each BS implementing the FD-JCAS scheme, communicates with its associated user in the DL direction, when its radar system is waiting for the object's echo (Fig. \ref{TimeSlot} $(b)$). We assume that the users are located over a circle of radius $d$ around their associated BSs, where $d\leq R$. For the considered highly directional mmWave network deployments, where there exists a form of repulsion among the networks nodes, the interference caused by nearby cells, and even more the interference caused by the reflected signals that travel large distances, is significantly reduced \cite{LI,LU,KON2,KON,DEN}. Therefore, the main terms that prevail in the aggregate network interference are caused by the nearby BSs that perform communication operations, allowing to neglect the interference caused by the reflected radar signals for simplicity. The validity of the above-mentioned assumption will be shown in the numerical results in Section \ref{Section6}. Let $\gamma^{\rm D}_k(d)$ represents the DL SINR observed by a user at a distance $d$ from its serving BS, that is evaluated as 
\begin{equation*}
\gamma^{\rm D}_k(d) = \frac{{P}_k(d) G^2{g}_0(\tau) L(d)}{\sum\nolimits_{i=1}^K\sum\nolimits_{X_{j,i}\in\widetilde{\Phi}_i}P_i(d)G^2\ell {g}_{j,i}(\tau) L\left(\|X_{j,i}\|\right)+\sigma_n^2},
\end{equation*}  
where $j\in\mathbb{N}$, ${g}_0(\tau)\sim\Gamma[\nu,1/\nu]$ and ${g}_{j,i}(\tau)\sim\Gamma[\nu,1/\nu]$ is the channel power gain during time slot $\tau$ of the user with its serving BS and the interfering BS at $X_{j,i}\in\widetilde{\Phi}_i$, respectively, and ${P}_k(d)$ denotes the transmission power allocated to a BS that belongs in the $k$-th tier, based on the power control mechanism. Following Alzer's Lemma \cite{BLA} in order to approximate the Gamma r.v. with a weighted sum of the cumulative distribution functions of exponential r.v., a tight upper bound of the DL coverage probability of a user i.e., $\mathbb{P}[\gamma^{\rm D}_k(d)>\eta]$, is given in the following proposition.

\begin{proposition}
	The DL coverage probability of a user at distance $d$ from its serving BS at $X_{i,k}\in\Phi_k$, is given by \eqref{EquationC} at the top of this page where $s=\eta\xi\varpi(\epsilon+d^a)/(P_k(d)G^2)$, $\varpi=\nu(\nu!)^{-\frac{1}{\nu}}$ and $f_{C_{i,k}}(u)$ represents the pdf of the $\beta$-GPP distributed interfering BSs, which is given in Proposition \ref{FirstLemma} with $c =\pi\widetilde{\lambda}_{k}(r)$.
\end{proposition}
The investigation of random distances between a BS and its associated user, which is achieved by un-conditioning the above expression with the pdf $f(d)=\frac{2}{R^2}d$, is an interesting but trivial extension, and thus is out of the scope of this work. The DL coverage performance in the context of $\beta$-GPP deployments, is a well-investigated metric within the literature \cite{BLA,KON,DEN}. Therefore, the focus of this paper lies in evaluating the achieved detection performance of the FD-JCAS systems for the considered network deployments, where the communication operation is indicated by the existence of interference at the radar operation of the BSs and residual SI at the JCAS-enabled BSs.
\section{Cooperative Multi-point Radar Detection Technique}\label{Section3}
In this section, we introduce the proposed CoMRD technique in the context of mmWave HetNets. Our technique exploits the cooperation among randomly located BSs that belong in different network tiers, aiming at enhancing the network’s detection performance. Specifically, our proposed technique is based on a two-level procedure: $(i)$ pre-selection phase, and $(ii)$ decision phase, which are detailed described in the following discussion. In addition, analytical expressions are derived which will be useful for evaluating the network’s detection performance in Section \ref{Section4}. 

\subsection{Pre-selection Phase}
We denote by $\mathcal{C}$, the set of cooperative BSs that are selected based on the rules of the adopted association scheme and share their sensing information to the FC for the collaborative detection of the presence or absence of an object at a PoI, where $|\mathcal{C}|\in\mathbb{N}$. Conventional distance-based (DB) association schemes  assume the existence of a centralized controller that selects the $K$ closest BSs from each PoI over the entire HetNet. However, such approaches require continuous exchange of information between the centralized controller and the BSs of the HetNet, resulting in increased implementation and computational complexity. Motivated by this, we consider a low-complexity per-tier DB (PTDB) association scheme \cite{HE,SKO2}, whose performance is a lower bound of the one achieved with the DB association scheme; this is shown in the numerical results of Section \ref{Section6}. Specifically, the PTDB scheme selects the closest BS to the typical PoI from each network tier $k\in\{1,\dots\,K \}$, which is denoted as $X^*_k\in{\Phi}_k$. In contrast to the DB association policy, the adopted scheme requires a signaling only between the BSs within the same network tier, reducing the signaling overhead and thus alleviating the computational complexity of the association process. This pre-selection policy requires an a priori knowledge of the location of the BSs, which can be obtained by monitoring the location of the BSs through a low-rate feedback channel or by a global positioning system mechanism \cite{SKO}. Hence, the set defined by the pre-selected BSs from all network tiers, represents the set of cooperative BSs for the object detection, and is defined as
\setcounter{equation}{3}
\begin{equation}
\mathcal{C}= \left\lbrace X^*_k :X^*_k= \argmin_{j\in\mathbb{N}}\|X_{j,k}\|,k\in\left\lbrace 1,\dots,K \right\rbrace \right\rbrace ,
\end{equation}
where $X^*_k$ represents the location of the closest BS to the PoI from the $k$-th tier. It is worth mentioning that, the adopted pre-selection policy is ideal for scenarios with low mobility, since it does not consider instantaneous fading.  

\subsection{Decision Phase}
At the second-level phase, each cooperative BS takes an independent binary decision regarding the presence of an object at the PoI based on the observed SINR. Let $\mathcal{H}_0$ and $\mathcal{H}_1$ represent the hypotheses made by each cooperative BS, when it senses the absence and the presence of an object at the typical PoI, respectively. In particular, the cooperative BS at $X^*_k\in\mathcal{C}$ from the $k$-th tier selects the hypothesis $\mathcal{H}_1$, when the observed SINR exceeds a pre-defined threshold $\vartheta$ i.e., $\gamma_{k}^{(\tau)}(\|X^*_k\|)\!\geq\!\vartheta$; otherwise, the cooperative BS selects the hypothesis $\mathcal{H}_0$. Hence, the binary decision of the cooperative BS at $X^*_k$, is defined as $\delta_k = 0$ if $\mathcal{H}_0\ {\rm holds}$, or $\delta_k = 1$ if $\mathcal{H}_1\ {\rm holds}$, where $k\!\in\!\left\lbrace 1,\dots,K\right\rbrace $. Thereafter, all local sensing informations from the cooperative BSs are shared to the FC via an ideal report channel. However, the distance between every BS and the typical PoI is different, thus, the quality of the sensing information of every BS is also different. Hence, in our proposed scheme, for each sensing information, a combining weight is assigned that reflects the quality of that local sensing information. Specifically, the weight of the cooperative BS at $X^*_k\in\mathcal{C}$, can be defined as,
\begin{equation}
w_k =\mathds{1}\left(\|X^*_k\|^{-1}\Big/\max\limits_{X^*_j\in\mathcal{C}}\left(\|X^*_j\|^{-1}\right)>\varsigma\right),
\end{equation}
where $i\in\left\lbrace 1,\dots,\eta\right\rbrace $ and $\varsigma$ represents a pre-defined fraction, with $0<\varsigma\leq1$. Thus, the final decision takes into account only the cooperative BSs whose distances from the typical PoI are relatively equal to that of the closest cooperative BS. The remaining cooperative BSs are considered to have a lower sensing quality due to their larger distances from the typical PoI, and therefore, have a zero impact on the final decision. Such weighted decision rule requires an a priori knowledge of the location of the BSs. This knowledge can be obtained by monitoring the location of the BSs through a low-rate feedback channel or by a global positioning system mechanism \cite{CHO}.

In order to enhance the overall detection capability of the considered network deployment, the FC makes the final decision by combining all weighted sensing information based on an adopted hard-decision combining rule. Specifically, the FC makes its decision according to the number of BSs claiming the presence of an object at their associated PoI. We adopt a generic $\kappa$-out-of-$K$ combining rule, where the FC decides that an object is present at the typical PoI, if and only if the $\kappa$ or more cooperative BSs decide the hypothesis $\mathcal{H}_1$. Hence, the final decision of the FC at the time slot $\tau$, conditioned on the cooperative BSs' locations, is given by
\begin{align}\label{GeneralRule}
&\mathcal{P}_d^{(\tau)}(\vartheta,\kappa,\varsigma|{\bf r}_{X^*})\nonumber\\&\!=\!  \sum\limits_{t=\kappa}^{K}\!\binom{K}{t}\!\prod\limits_{k=1}^{K}\!w_k\mathcal{P}_{d,k}^{(\tau)}\!(\vartheta|r_{X^*_k})^{\delta_k}\!\left(\!1\!-\!\mathcal{P}_{d,k}^{(\tau)}(\vartheta|r_{X^*_k})\!\right)^{1-\delta_k},
\end{align}
where $\kappa\in\mathbb{N}^+$, ${\bf r}_{X^*} = \{r_{X^*_1},\dots,r_{X^*_K} \}$ represents the set of the distances between the typical PoI and the cooperative BS from each tier, and $\mathcal{P}_{d,k}^{(\tau)}(\vartheta|r_{X^*_k})$ is the conditional detection probability achieved by the cooperative BS that belongs to the $k$-th tier, which according to the discussion in Section \ref{RadSection}, is evaluated as $\mathcal{P}_{d,k}^{(\tau)}(\vartheta|r_{X^*_k}) = \zeta p_{\rm L}\mathbb{P}\left[\gamma_k^{(\tau)}(r_{X^*_k})\geq\vartheta|r_{X^*_k}\right]$ \cite{PAR,HOU}.  Note that, $\delta_k$ is the detection decision of the cooperative BS at $X^*_k\in\mathcal{C}$ for the considered combination. It is important to mention here that, if $\kappa=\ceil{\frac{K}{2}}$, the rule is referred as the \textit{Majority rule}, if $\kappa=1$, the rule is referred as the \textit{OR rule}, and if $\kappa=K$, the rule is referred as the \textit{AND rule}. 

In contrast to the non-cooperative techniques \cite{PAR,REN}, where the detection decision is performed locally at each BS, our proposed cooperative technique combines all local decisions conducted by the cooperative BSs at the FC, aiming to determine the final sensing result \cite{CHO}. Hence, the need for a FC, makes our proposed technique more costly in terms of information exchange overhead compared to the non-cooperative counterpart. Moreover, the adopted association scheme for the pre-selection of the $K$ cooperative BSs offers lower implementation complexity and reduced signaling overhead compared to the conventional approaches.

\section{Detection Performance with CoMRD technique}\label{Performance}\label{Section4}
In this section, we analyze the impact of our proposed technique on the detection performance of FD-JCAS-enabled mmWave HetNets under a constant false alarm rate (CFAR) constraint, where the BSs from the $k$-th tier are deployed based on a $\beta$-GPP. Initially, the detection threshold is evaluated aiming at the achievement of the desired probability of FA $P_{\rm fa}$ for each tier $k$, where $k\in\{1,\dots,K\}$. Then, the SINR complementary cumulative distribution function (ccdf) is defined, and the Laplace transform of the overall interference function is characterized. By using the proposed CoMRD technique, the detection performance of each hard-decision rule is derived analytically by using stochastic geometry tools.

\subsection{Interference Characterization}
Firstly, we investigate the received interference at a cooperative BS, where analytical and asymptotic expressions for the Laplace transform of the received interference are derived. The aggregate interference at a cooperative BS located at $x\in\Phi_k$, can be expressed as follows
\begin{equation*}
\mathcal{I}=\sum\nolimits_{z=1}^K\sum\nolimits_{X_{j,z}\in\widetilde{\Phi}_z}P_i(d)G^2\ell {h}_{j,z}(\tau)L(\|X_{j,z}-x\|),
\end{equation*}
where $j\in\mathbb{N}$. Due to the non-orthogonal operation of communication and radar systems, the cooperative BS at $x\in\Phi_k$ may also observe a residual SI, where after the SI cancellation, is equal to 
\begin{equation*}
I_{\rm SI}(d,k) = P_k\left(d\right)g_{\rm SI}=\min\{\rho L(d)^\varepsilon,\ P_k\}g_{\rm SI},
\end{equation*}
where $g_{\rm SI}$ represents the power gain of the residual SI channel. To characterize the interference, in the following Lemma, we compute the Laplace transform of the r.v. $\mathcal{I}$ and $I_{\rm SI}$ evaluated at $s$.

\begin{lemma}\label{LaplaceInterfereneLemma}
	The Laplace transform of the aggregate interference during time slot $\tau$, $\mathcal{L}^{(\tau)}_{\mathcal{I}}\left(s\right)$, is given by
	\begin{equation}\label{LaplaceInterference}
	\mathcal{L}^{(\tau)}_{\mathcal{I}}\left(s\right)\!=\!\prod_{z=1}^K\!\prod\limits_{i\in\mathbb{N}\backslash j}\!\left(\!1\!-\!\beta\!+\!\int_{0}^{\sqrt{{R}}}\frac{\beta f_{C_{i,z}}(y)}{1+\frac{s P_z(d)G^2\ell}{\left(\epsilon+y^\frac{a}{2}\right)}}dy \right),
	\end{equation}
	and the Laplace transform of the SI, $\mathcal{L}_{I_{\rm SI}}\left(s\right)$, can be expressed as
	\begin{equation}\label{LaplaceSI}
	\mathcal{L}_{I_{\rm SI}}\left(s\right)=\left(\frac{\mu}{\mu+s\sigma^2_{\rm SI}\min\{\rho \left(\epsilon+d^{a}\right)^\varepsilon,P_k\}}\right)^\mu,
	\end{equation}
	where $f_{C_{i,k}}(u)$ represents the pdf of the $\beta$-GPP distributed interfering BSs, which is given in Proposition \ref{FirstLemma} with $c = \pi\widetilde{\lambda}_{k}$.
\end{lemma}
\begin{proof}
	Firstly, we compute the Laplace transform of r.v. $\mathcal{I}$ at $s$ during time slot $\tau$, conditioned on the location of the serving BS $X_{i,k}$, which we denote as $\mathcal{L}^{(\tau)}_\mathcal{I}\left(s\right)$. The Laplace transform definition yields
	\begin{align}
	&\mathcal{L}^{(\tau)}_{\mathcal{I}}\left(s\right) = \mathbb{E}\left[e^{-s\left(\sum\limits_{z=1}^K\sum\limits_{i\in\mathbb{N}\backslash j}P_z(d)G^2\ell h_{i,z}(\tau) L\left({C}_{i,z}^\frac{1}{2}\right)\Xi_{i,z}\right)}\right]\nonumber\\
	&=\!\prod\limits_{z=1}^K\!\prod\limits_{i\in\mathbb{N}\backslash j}\!\!\left(\!\beta_z \mathbb{E}\left[e^{-s P_z(d)G^2\ell h_{i,z}(\tau) L\left({C}_{i,z}^\frac{1}{2}\right)}\right]\!+\!1\!-\!\beta_z\!\right)\label{Proof1}\\
	&=\!\prod\limits_{z=1}^K\!\prod\limits_{i\in\mathbb{N}\backslash j}\!\left(\!1\!-\!\beta_z\!+\!\mathbb{E}\left[\frac{\beta_z}{1\!+\!s P_z(d)G^2\ell L\left({C}_{i,z}^\frac{1}{2}\right)} \right]\right)\label{Proof2}
	\end{align}
	where \eqref{Proof1} is derived based on the $\beta$-GPP properties, and \eqref{Proof2} is due to the moment generating function of an exponential r.v.. By using the distance distribution of $\beta$-GPP deployed BSs, which is calculated in Proposition \ref{CPpdf}, \eqref{LaplaceInterference} can be derived. Regarding the expressions for the Laplace transform of the SI function, it can be expressed as
	\begin{align}
	\mathcal{L}_{I_{\rm SI}}\left(s\right)=\mathbb{E}\left[\exp\left(-s \min\{\rho \left(\epsilon+d^{a}\right)^\varepsilon,P_k\}g_{\rm SI}\right)\right],
	\end{align}
	where $g_{\rm SI}$ is the residual SI channel that follows Gamma distribution i.e., $g_{\rm SI}\sim\mathcal{G}\left(\mu,\frac{\sigma_{\rm SI}^2}{\mu}\right)$. By averaging over the channel fading, the final expression for the Laplace transform of the SI can be derived
\end{proof}

We end this section by computing the Laplace transform of the aggregate interference $\mathcal{I}$ in the case of PPP (which is obtained as the limit as $\beta_k\rightarrow 0$ in the Lemma \ref{LaplaceInterfereneLemma}). 

\begin{corollary}\label{Cor1}
	Let $\widetilde{\Phi}_k\sim$PPP($\widetilde{\lambda}_{k}(r)$), the Laplace transform of $\mathcal{I}$ during time slot $\tau$, $\widetilde{\mathcal{L}}^{(\tau)}_{\mathcal{I}}\left(s\right)$, is given by
	\begin{align}\label{LaplaceInterferencePPP}
	\!\!\widetilde{\mathcal{L}}^{(\tau)}_\mathcal{I}\left(s\right)&\!=\!\prod_{z=1}^{K}\!e^{-\lambda_{z}p_{\rm L}\frac{\phi}{2}{R}^2\!\left(1-\frac{\epsilon}{Q(s,z)}\right){}_2F_1\left[1,\frac{2}{a},\frac{2+a}{a},-\frac{{R}^a}{Q(s,z)}\right]},
	\end{align}
	where $Q(s,i) = G^2\ell P_i s+\epsilon$.
\end{corollary}

\subsection{Detection Threshold Under False Alarm Constraint}
The detection threshold is tuned according to the principle of the CFAR \cite{REN}. Specifically, the threshold is designed (a priori) in such a way as to achieve a desired false alarm rate ($P_{\rm fa}$) over all network realizations. Moreover, this design takes into account the topology and the statistical properties of the network (e.g., interference, network density, etc). Therefore, even though this threshold is constant, it is evaluated based on a contextual-aware design.

In order to derive compact and insightful expressions for the FA probability, and hence, the detection performance, the detection threshold is evaluated for the case of PPP-distributed BSs, i.e. $\beta_k\rightarrow 0,\ \forall k\in\{0,\dots,K\}$, with omnidirectional antenna gains, i.e. $\phi=2\pi$, and $a=4$. These assumptions provide an approximation of the actual network's detection performance, which is shown to be tight from numerical results in Section \ref{Section6}. Hence, we assume that the required FA rate for the $k$-th tier, is equal to $P_{\rm fa}=(1-p_{\rm L})\zeta\mathbb{P}[\gamma_k^{(\tau)}(u)>\vartheta]$, where $\gamma_k^{(\tau)}(u)$ is the signal reflected at a blockage at distance $u$ from the cooperative BS that belongs in the $k$-th tier, and is given by \eqref{SINR}. In the following lemma, the $P_{\rm fa}$ is analytically evaluated.

\begin{lemma}\label{LemmaFA}
	The FA probability for the $k$-th tier of PPP-distributed omnidirectional HetNets with $a=4$, is given by
	\begin{equation}
	P_{\rm fa} = (1-p_{\rm L})\zeta\Big/\left(1+p_{\rm L}\sqrt{\pi\vartheta}\frac{\lambda_k}{\lambda_b}\right).
	\end{equation}
\end{lemma}
\begin{proof}
	Based on the discussion in Section \ref{RadSection}, the FA probability for the $k$-th tier, can be re-written as
	\begin{align}
	&P_{\rm fa}=(1-p_{\rm L})\zeta\int_{0}^{R}\mathbb{P}\left[\frac{P_{\rm refl}^{(\tau)}(u)}{\mathcal{I}}>\vartheta\right]f_b(u)\diff u\nonumber\\
	&=\!(1\!-\!p_{\rm L})\zeta\int_{0}^Re^{-\pi\lambda_{k} p_{\rm L}R\sqrt{P_ks\ell}\arctan\left[\frac{R}{\sqrt{P_k s\ell}}\right]}f_b(u)\diff u,\label{P1}
	\end{align}
	where $f_b(u)$ is the probability density function (pdf) of the distance between a cooperative BS from the $k$-th tier and its closest blockage, i.e. $f_b(u)=2\pi\lambda_bu\exp(-\pi\lambda_bu^2)$ \cite{PAR,AND}, and \eqref{P1} is derived by using the expression in Corollary \ref{Cor1} under the considered assumptions. Finally, by assuming $R\rightarrow\infty$, the final expression can be obtained.
\end{proof}

Leveraging the expressions derived in Lemma \ref{LemmaFA}, the detection threshold needed for achieving the desired FA rate is analytically computed in the following lemma.
\begin{lemma}\label{Threshold}
	The detection threshold $\theta_k$ for the $k$-th tier, which ensures a FA rate $P_{\rm fa}$, is equal to
	\begin{equation}\label{DetectionThreshold}
	\theta_k = \frac{A\lambda_b^2}{\lambda_{k}^2p_{\rm L}\pi^3R^2}\left(1-\frac{(1-p_{\rm L})\zeta}{P_{\rm fa}}\right)^2.
	\end{equation}
\end{lemma}
\begin{proof}
	The expression is derived by solving the expression derived in Lemma \ref{LemmaFA} with respect to the detection threshold.
\end{proof}
The simple closed-form expression in \eqref{DetectionThreshold} provides guidance on how to tune the radar receiver for network tier, pinpointing the impact of all relevant system parameters. 

\subsection{Detection Performance under CFAR Constraint}
Based on Section \ref{bGPP} and by following Slivnyak's theorem, the effective SINR of the detection process for the cooperative BS $X^*_{k}\in\Phi_k$, can be expressed as 
\begin{align}
&\gamma_k^{(\tau)}\!(\widetilde{B}_{j,k}^\frac{1}{2})\!=\!\frac{P_kG\frac{A\ell}{4\pi}h_{0}(\tau) L^2\left(\widetilde{B}_{j,k}^\frac{1}{2}\right)}{\sum\limits_{z=1}^K\sum\limits_{i\in\mathbb{N}\backslash j}\!\!\frac{P_z(d)G^2\ell {h}_{i,z}(\tau)}{ L\left(\widetilde{C}_{i,z}^\frac{1}{2}\right)^{-1}}\!+\!I_{\rm SI}\!\left(d,k\right)\mathds{1}_{ \text{JCAS}}\!+\!\sigma_n^2},\nonumber\\
&\quad\!=\!\frac{P_kG\frac{A\ell}{4\pi}h_{0}(\tau) L^2\left({B}_{j,k}^\frac{1}{2}\right)\Xi_{j,k}}{\sum\limits_{z=1}^K\sum\limits_{i\in\mathbb{N}\backslash j}\!\!\frac{P_z(d)G^2\ell {h}_{i,z}(\tau)\Xi_{i,z}}{ L\left({C}_{i,z}^\frac{1}{2}\right)^{-1}}\!+\!I_{\rm SI}\!\left(d,k\right)\mathds{1}_{ \text{JCAS}}\!+\!\sigma_n^2},
\end{align}
where $\Xi_{j,k}$ denotes the discrete r.v. of the $j$-th BS from the $k$-th tier, i.e $\Xi_{j,k}=\{0,1\}$. The following lemma provides the conditional detection probability under CFAR constraint in the context of $\beta$-GPP-distributed mmWave HetNets.

\setcounter{equation}{17}
\begin{figure*}[t]
	\begin{align}
	\mathcal{P}_{d,k}^{(\tau)}(\theta_k|B_{j,k})
	&=\chi_k\zeta p_{\rm L}\mathbb{P}\left[h_0(\tau)\geq s\left(\sum\limits_{z=1}^K\sum\limits_{i\in\mathbb{N}\backslash j}P_z(d)G^2\ell {h}_{i,z}(\tau) L\left({C}_{i,z}^\frac{1}{2}\right)+I_{\rm SI}\left(d,k\right)+\sigma_n^2\right)|\Xi_{j,k}=1\right]\mathbb{P}[\Xi_{j,k}=1]\nonumber\\
	&\qquad\qquad\qquad\ +(1\!-\!\chi_k)\zeta p_{\rm L}\mathbb{P}\!\left[\!h_0(\tau)\!\geq\! s\!\left(\!\sum\limits_{z=1}^K\!\sum\limits_{i\in\mathbb{N}\backslash j}\!P_z(d)G^2\ell {h}_{i,z}(\tau) L\!\left({C}_{i,z}^\frac{1}{2}\right)\!+\!\sigma_n^2\right)\!|\Xi_{j,k}\!=\!1\right]\!\mathbb{P}[\Xi_{j,k}\!=\!1]\label{Alzer},
	\end{align}
	\hrulefill
	\begin{align}
	\mathcal{P}_{d,k}^{(\tau)}(\theta_k|B_{j,k}) &\leq\beta_k\zeta p_{\rm L}\sum_{\xi=1}^\nu(-1)^{\xi+1}\binom{\nu}{\xi}\Bigg(\chi_k\mathbb{E}\left[\exp\left(- s\xi\varpi\left(\sum\limits_{z=1}^K\sum\limits_{i\in\mathbb{N}\backslash j}P_z(d)G^2\ell {h}_{i,z}(\tau) L\left({C}_{i,z}^\frac{1}{2}\right)+I_{\rm SI}\left(d,k\right)+\sigma_n^2\right)\right)\right]\nonumber\\
	&\qquad\qquad\qquad\qquad\qquad+(1-\chi_k)\mathbb{E}\left[\exp\left( -s\xi\varpi\left(\sum\limits_{z=1}^K\sum\limits_{i\in\mathbb{N}\backslash j}P_z(d)G^2\ell {h}_{i,z}(\tau) L\left({C}_{i,z}^\frac{1}{2}\right)+\sigma_n^2\right)\right)\right]\Bigg),\label{BoundEq}
	\end{align}
	\hrulefill
\end{figure*}
\setcounter{equation}{16}

\begin{lemma}\label{Lemma1}
	The conditional detection probability achieved by the cooperative BS at $X^*_{k}\in{\Phi}_k$ at the time slot $\tau$, where the BSs are distributed according to a $\beta$-GPP, is given by
	\begin{align}
	\mathcal{P}_{d,k}^{(\tau)}(\theta_k|B_{j,k})&\leq\beta_k\sum_{\xi=1}^\nu(-1)^{\xi+1}\binom{\nu}{\xi}\exp\left(-s\xi\varpi\sigma_n^2\right)\nonumber\\&\times\mathcal{L}_\mathcal{I}^{(\tau)}(s\xi\varpi)\left(\chi_k\mathcal{L}_{I_{\rm SI}}\left(s\xi\varpi\right)+1-\chi_k\right),
	\end{align}
	where $s = \frac{4\pi\theta_k(\epsilon+B_{j,k}^\frac{a}{2})^2}{P_kGA\ell}$, $\varpi=\nu(\nu!)^{-\frac{1}{\nu}}$, $f_{C_{i,k}}(u)$ represents the pdf of the $\beta$-GPP distributed interfering BSs, which is given in Proposition \ref{FirstLemma}, and $\mathcal{L}_\mathcal{I}^{(\tau)}\left(\alpha\right)$ and $\mathcal{L}_{I_{\rm SI}}\left(\alpha\right)$ represent the Laplace transform of the overall interference function and the SI, respectively, and are given in Lemma \ref{LaplaceInterfereneLemma}
\end{lemma}
\begin{proof}
	\setcounter{equation}{19}
	Conditioned on whether the considered cooperative BS $X^*_k$ performs the FD-JCAS scheme or not, its achieved detection performance, conditioned on its location, is given by \eqref{Alzer} at the top of next page, where $s=\frac{4\pi\theta_k(\epsilon+B_{j,k}^\frac{1}{2})^\frac{a}{2}}{P_kG^2A\ell}$ and $h_0(\tau)$ is a gamma r.v. i.e., $h_0(\tau)\sim\Gamma[\nu,1/\nu]$. To overcome the difficulty on Nakagami fading, Alzer's Lemma \cite{BLA} on the ccdf of a gamma r.v. with integer parameter can be applied. This relates the ccdf of a gamma r.v. into a weighted sum of the ccdfs of exponential r.v.. Hence, we can bound expression \eqref{Alzer} as in \eqref{BoundEq} at the top of next page, where $\varpi=\nu(\nu!)^{-\frac{1}{\nu}}$, $\mathcal{L}_\mathcal{I}^{(\tau)}(\alpha)$, and $\mathcal{L}_{I_{\rm SI}}(\alpha)$  are the Laplace transforms of the overall interference function and SI, respectively, that are derived in Lemma \ref{LaplaceInterfereneLemma}.
\end{proof}
It is noteworthy to mention that the Laplace transform of the overall interference function  in \eqref{LaplaceInterference} is a decreasing function of $\beta$. This can be explained by the fact that, the distance between the cooperative BSs and their interfering BSs becomes larger as $\beta\rightarrow 1$. Specifically, when $\beta$ becomes larger, the BSs exhibit more repulsion and tend to be distributed more uniformly. This leads to an increase of the distance between the cooperative BSs and their interfering BSs, resulting in a decreased overall interference. We now state our main result for the detection probability achieved by our proposed CoMRD technique in the context of $\beta$-GPP mmWave HetNets.

\begin{theorem}\label{Theorem1}
	The detection performance achieved by the CoMRD technique at the time slot $\tau$, is given by 
	\begin{align}
	\mathcal{P}_d^{(\tau)}(\boldsymbol{\theta},\kappa,\varsigma)&\leq \sum\limits_{t=\kappa}^{K}\binom{K}{t}\prod\limits_{k=1}^{K}\sum\limits_{j\in\mathbb{N}}w_k\int_{0}^{\sqrt{{R}}} \mathcal{P}_{d,k}^{(\tau)}(\theta_k|u)^{\delta_k}\nonumber\\&\quad\times\left(1-\mathcal{P}_{d,k}^{(\tau)}(\theta_k|u)\right)^{1-\delta_k}f_{B_{j,k}}(u)du,\label{NoCorrelation}
	\end{align}
	where $\boldsymbol{\theta} = \{\theta_1,\dots,\theta_K\}$, $s = \frac{4\pi\theta_k(\epsilon+B_{j,k}^\frac{a}{2})^2}{P_kG^2A\ell}$, $\mathcal{P}_{d,k}^{(\tau)}(\theta_k|u)$ represents the achieved detection performance of the cooperative BS $X_{j,k}\in{\Phi}_k$ at the time slot $\tau$, and $f_{B_{j,k}}(u)$ represents the pdf of the $\beta$-GPP BSs, which is given in Proposition \ref{FirstLemma}.
\end{theorem}
\begin{proof}
	By substituting the expression of Lemma \ref{Lemma1} in \eqref{GeneralRule}, and by un-conditioning the derived expression with the pdf given by \eqref{CPpdf}, the result in Theorem \ref{Theorem1} follows.
\end{proof}

To illustrate the effect of repulsion between the BSs on the achieved detection performance, we also evaluate our proposed technique in the context of PPP mmWave HetNets. As we previously mentioned, the independent deployment of BSs i.e., the BSs are spatially distributed based on a PPP, is a special case of the considered $\beta$-GPP deployment, where $\beta\rightarrow 0$. In the following corollary, we state the achieved detection performance of our proposed technique in the context of PPP mmWave HetNets. 
 \begin{corollary}
 	When $\beta\rightarrow 0$, the detection performance achieved with the CoMRD technique at the time slot $\tau$, is given by
 	\begin{align}
 	&\mathcal{P}_d^{(\tau)}(\theta_k,\kappa,\varsigma)\!\leq\! \sum\limits_{t=\kappa}^{K}\binom{K}{t}\prod\limits_{k=1}^{K}w_k\int_{0}^{{R}}\widetilde{\mathcal{P}}_{d,k}^{(\tau)}(\theta_k|u)^{\delta_k}\nonumber\\&\qquad \times\!(1\!-\!\widetilde{\mathcal{P}}_{d,k}^{(\tau)}(\theta_k|u))^{1-\delta_k}2\pi\lambda_k u \zeta p_{\rm L}e^{-\pi\lambda_{k}\zeta p_{\rm L}u^2}du,
 	\end{align}
 	where $\widetilde{\mathcal{P}}_{d,k}^{(\tau)}(\theta_k|u)$ represents the conditional detection probability, and is given by
 	\begin{align*}
 	\widetilde{\mathcal{P}}_{d,k}^{(\tau)}(\theta_k|u)\leq\sum_{\xi=1}^\nu(-1)^{\xi+1}\nu&(\nu!)^{-\frac{1}{\nu}}\widetilde{\mathcal{L}}_\mathcal{I}^{(\tau)}(s\xi\varpi)\exp\left(-s\xi\varpi\sigma^2_n\right)\\&\times\left(\chi_k\mathcal{L}_{I_{\rm SI}}\left(s\xi\varpi\right)+1-\chi_k\right),
 	\end{align*}
 	where $\widetilde{\mathcal{L}}_{\mathcal{I}}^{(\tau)}(\alpha)$ and $\mathcal{L}_{I_{\rm SI}}(\alpha)$ are given by \eqref{LaplaceInterferencePPP} and \eqref{LaplaceSI}, respectively, and $s = \frac{4\pi\theta_k(\epsilon+u^a)^2}{P_kG^2A\ell}$.
 	
 \end{corollary}

\section{Temporal Correlation of Radar Detection Performance}\label{Section5}
In this section, we assess the effect of the temporal interference correlation on the detection performance achieved by our proposed technique. Even though we assume that the channel fading between different time slots is independent, the interference caused at a certain location is correlated across the time for the same network realization due to the fixed locations of the BSs. Specifically, by assuming a static PoI, a fraction of the interfering BSs at a given time slot might also cause interference at the particular PoI in future time slots, which introduces a temporal interference correlation that needs to be taken into account. Hence, due to the temporal correlation of interference, if an object is successfully detected by its cooperative BSs in a given time slot, there is a higher probability that the object will also be detected in the future time slots \cite{KRI}.

\setcounter{equation}{22}
\begin{figure*}[t]
	\begin{align}
	\mathcal{P}_d^{(\tau,\widehat{\tau})}(\theta_k,\kappa,\varsigma|B_{j,k}) &\stackrel{(a)}{=}\beta_k\mathbb{P}\left[h_0(\tau)\geq s\left(\mathcal{I}(\tau)+I_{\rm SI}(R,k)\mathds{1}_\text{JCAS}+\sigma_n^2\right),h_0({\widehat{\tau}})\geq s\left(\mathcal{I}(\widehat{\tau})+I_{\rm SI}(R,k)\mathds{1}_\text{JCAS}+\sigma_n^2\right)\right],\label{Equation2}\\
	&\stackrel{(b)}{\leq}\!\beta_k\!\sum_{\xi=1}^K(-1)^{\xi+1}\!\binom{\nu}{\xi}\!\exp\left(-2s\xi\varpi\sigma_n^2\right)\!\left(\chi_k\big(\mathcal{L}_{I_{\rm SI}}(s\xi\varpi)\big)^2\!+\!1\!-\!\chi_k\right)\!\mathbb{P}\left[h_0(\tau)\geq s\mathcal{I}(\tau),h_0(\widehat{\tau})\geq s\xi\varpi\mathcal{I}(\widehat{\tau})\right]\nonumber\\
	& =\beta_k\exp\left(-2s\xi\varpi\sigma_n^2\right)\left(\chi_k\big(\mathcal{L}_{I_{\rm SI}}(s\xi\varpi)\big)^2+1-\chi_k\right)\mathcal{L}_\mathcal{I}^{(\tau,\widehat{\tau})}(s\xi\varpi).\nonumber
	\end{align}
	\hrulefill
\end{figure*}
\setcounter{equation}{21}

Consider a \textit{static object} scenario, where the object at the typical PoI remains static at the origin for multiple time slots. For this scenario, we study the ability of our proposed technique to successfully detect the object at the typical PoI over different time slots. Let $\mathcal{P}_d^{(\tau,\widehat{\tau})}(\theta_k,\kappa,\varsigma|X^*_{k})$ denote the probability of the cooperative BS $X^*_{k}\in{\Phi}_k$ to jointly detect its associated PoI at both time slots $\tau$ and $\widehat{\tau}$, which can be expressed as
\begin{equation}\label{TemporalDetection}
\mathcal{P}_d^{(\tau,\widehat{\tau})}\!(\theta_k,\kappa,\varsigma|X^*_{k})\! =\! \mathbb{P}\!\left[\gamma_{k}^{(\tau)}(r_{X^*_{k}})\!\geq\!\theta_k,\gamma_{k}^{(\widehat{\tau})}(r_{X^*_{k}})\!\geq\!\theta_k\right],
\end{equation}
where $\tau\!\neq\!\widehat{\tau}$. In the context of $\beta$-GPP mmWave HetNets, this probability can be expressed as in \eqref{Equation2} at the top of next page, where $s=4\pi\theta_k\left(\epsilon+B_{j,k}^\frac{a}{2}\right)^2/(P_kG^2A\ell)$, $\varpi=\nu(\nu!)^{-\frac{1}{\nu}}$,$\mathcal{I}(t)=\sum\nolimits_{z=1}^K\sum\nolimits_{i\in\mathbb{N}\backslash j}P_z(d)G^2\ell {h}_{i,z}(t) L\left(\|C_{i,z}^\frac{1}{2}\|\right)$ and $\mathcal{L}_{\mathcal{I}}^{(\tau,\widehat{\tau})}(\alpha)$ represents the joint Laplace transform of the interference functions at the time slots $\tau$ and $\widehat{\tau}$, where $\tau\neq\widehat{\tau}$, and is equal to $\mathcal{L}_\mathcal{I}^{(\tau,\widehat{\tau})}(\alpha)=\mathbb{P}[h_0(\tau)>\alpha\mathcal{I}(\tau),h_0(\widehat{\tau})>\alpha\mathcal{I}(\widehat{\tau})]$; $(a)$ is derived by substituting \eqref{SINR} for different time instances $\tau$ and $\widehat{\tau}$ and due to the independence between $h_0(\tau)$ and $h_0(\widehat{\tau})$, and $(b)$ is based on Alzer's lemma. In the following lemma, the expression for the joint Laplace transform of the interference functions at the time slots $\tau$ and $\widehat{\tau}$ is evaluated.
\setcounter{equation}{23}
\begin{lemma}\label{Lemma2}
	The joint Laplace transform of the interference functions at the time slots $\tau$ and $\widehat{\tau}$, is given by
	\begin{equation}\label{Proof7}
	\mathcal{L}_\mathcal{I}^{(\tau,\widehat{\tau})}(s)\!=\! \int\limits_{0}^{\sqrt{{R}}}\!\prod\limits_{z=1}^{K}\!\prod\limits_{i\in\mathbb{N}\backslash j}\!\!\left(\!1\!-\!\beta\!+\!\frac{\beta }{1\!+\!\frac{\alpha P_z(d)G^2\ell}{ \epsilon+y^{\frac{a}{2}}}}\! \right)^2\!\!f_{C_{i,z}}\!(y)dy,
	\end{equation}
	where $f_{C_{i,z}}(u)$ is the pdf of the $\beta$-GPP distributed interfering BSs, which is given in Proposition \ref{FirstLemma} with $c = \pi\widetilde{\lambda}_{z}$.
\end{lemma}
\begin{proof}
Initially, since $h_0(\tau)$ and $h_0({\widehat{\tau}})$ are independent for $\tau\neq\widehat{\tau}$, the joint Laplace transform of the interference functions at the time slots $\tau$ and $\widehat{\tau}$, is given by
\begin{align}
&\mathcal{L}_\mathcal{I}^{(\tau,\widehat{\tau})}(\alpha)=\mathbb{E}\!\left[\!e^{\!-\alpha\!\sum\limits_{z=1}^K\sum\limits_{i\in\mathbb{N}\backslash j}\!P_z(d)G^2\ell h_{i,z}(\tau) L\!\left(\!{B}_{i,z}^\frac{1}{2}\!\right)\Xi_{i,z}}\right.\nonumber\\
&\quad\quad\quad\quad\times\left.e^{-\alpha\sum\limits_{z=1}^K\sum\limits_{i\in\mathbb{N}\backslash j}\!P_z(d)G^2\ell g_{i,z}(\widehat{\tau}) L\!\left(\!{B}_{i,z}^\frac{1}{2}\!\right)\Xi_{i,z}}\right]\nonumber\\
&=\mathbb{E}\!\left[\!\prod\limits_{z=1}^K\prod\limits_{i\in\mathbb{N}\backslash j}\!\left(\!1\!-\!\beta_z\!+\!\frac{\beta_z}{1\!+\!\alpha P_z(d)G^2\ell L\left({B}_{i,z}^\frac{1}{2}\right)} \right)^2\right]\label{Proof4}
\end{align}
where \eqref{Proof4} follows from the moment generating functional of the exponential r.v. $h_0(\tau)$ and $h_0(\widehat{\tau})$, and by un-conditioning the expression with the pdf $f_{C_{i,z}}(y)$ (Proposition \ref{Lemma1}), the final expression can be derived.
\end{proof}
By substituting \eqref{Proof7} in \eqref{TemporalDetection}, and the resulting expression in \eqref{GeneralRule}, and by un-conditioning the derived expression with the pdf given by \eqref{CPpdf}, the joint detection probability for the time slots $\tau$ and $\widehat{\tau}$ can be expressed as
\begin{align}
\mathcal{P}_d^{(\tau,\widehat{\tau})}&(\theta_k,\kappa\varsigma)\!\leq\! \sum\limits_{t=\kappa}^{K}\!\binom{K}{t}\!\prod\limits_{k=1}^{K}\!\left(\sum\limits_{j\in\mathbb{N}}w_k\int\limits_{0}^{\sqrt{{R}}}\! \mathcal{P}_{d,k}^{(\tau,\widehat{\tau})}(\theta_k\kappa,\varsigma|u)^{\delta_k}\right.\nonumber\\&\quad\ \ \times\left.\left(1\!-\!\mathcal{P}_{d,k}^{(\tau,\widehat{\tau})}(\theta_k,\kappa,\varsigma|u)\right)^{1-\delta_k}\!f_{B_{j,k}}(u)du\!\right),\label{Correlation}
\end{align}
where $f_{B_{j,k}}(u)$ denotes the distance distribution between the cooperative BS $X^*_{k}$ and the typical PoI. The following theorem characterizes the detection probability at the time slot $\widehat{\tau}$, conditioned on the detection probability at the time slot $\tau$, in the context of our proposed technique and in a $\beta$-GPP deployment.
\begin{theorem}
	The achieved detection probability with the CoMRD technique at the time slot $\widehat{\tau}$, conditioned on the detection probability at the time slot $\tau$, is given by
	\begin{equation}
	\mathcal{P}_d^{(\widehat{\tau}|\tau)}(\theta_k,\kappa,\varsigma) = \frac{\mathcal{P}_d^{(\tau,\widehat{\tau})}(\theta_k,\kappa,\varsigma)}{\mathcal{P}_d^{(\tau)}(\theta_k,\kappa,\varsigma)},
	\end{equation}
	where $\mathcal{P}_d^{(\tau)}(\theta_k,\kappa,\varsigma)$ and $\mathcal{P}_d^{(\tau,\widehat{\tau})}(\theta_k,\kappa,\varsigma)$ are given by \eqref{NoCorrelation} and \eqref{Correlation}, respectively.
\end{theorem}

In order to illustrate the effect of temporal correlation on the detection performance, we introduce the ratio of the conditional and the unconditional detection probability, that is given by
\begin{align}\label{Ratio}
\varrho &= \frac{\mathcal{P}_d^{(\widehat{\tau}|\tau)}(\theta_k,\kappa,\varsigma)}{\mathcal{P}_d^{(\tau)}(\theta_k,\kappa,\varsigma)}=\frac{\mathcal{P}_d^{(\tau,\widehat{\tau})}(\theta_k,\kappa,\varsigma)}{\mathcal{P}_d^{(\tau)}(\theta_k,\kappa,\varsigma)^2}>1.
\end{align}
\begin{remark}
	From \eqref{Ratio}, if the detection process succeeds at the time slot $\tau$, there is a higher probability that the detection process succeeds at future time slots $\widehat{\tau}$.
\end{remark}
\begin{remark}
	From \eqref{Ratio}, we can easily observe that $\mathcal{P}_d^{(\widehat{\tau}|\tau)}(\theta_k,\kappa,\varsigma)>\mathcal{P}_d^{(\tau)}(\theta_k,\kappa,\varsigma)$. Particularly, a failure in the object detection by the CoMRD technique at the time slot $\tau$, results in a more likely detection failure at future time slots $\widehat{\tau}$. 
\end{remark}
\section{Numerical Results}\label{Section6}
In this section, we plot the analytical results derived in the previous sections along with the simulation results obtained from Monte-Carlo trials. We focus on the special case of a HetNet with $K = 3$ tiers, where their spatial densities are $\lambda_1 = 1\ {\rm BSs/km^2}$, $\lambda_2 = 2\ {\rm BSs/km^2}$, and $\lambda_3 =  4\ {\rm BSs/km^2}$, respectively, with transmit power equal to $P_1 = 15$ dBm, $P_2 = 10$ dBm, and $P_3= 5$ dBm, respectively \cite{AND}. Regarding the repulsive behavior between the BSs, we assume that the repulsive parameter of all tiers is equal to $\beta=0.9$ i.e., $\beta_k=\beta=0.9\ \forall k$. For the large-scale path-loss, we assume $a = 4$ and $\epsilon=1$, while the Nakagami fading parameter is equal to $\nu=2$. The power control factor is $\varepsilon=0.9$ and all BSs have the same receive sensitivity $\rho=-40$ dB \cite{SAK}. The SI capabilities of BSs are set to $\sigma^2_{\rm SI}=-60$ dB and $\mu=4$ \cite{SAK}. The parameters for the sectorized antenna model are set to $G = 10$ dB and $\phi = \frac{\pi}{6}$ for the main lobe gain and the main lobe beamwidth, respectively. In terms of modeling the blockage sensitivity of mmWave signals, we assume that within a disk of radius ${R}=400$ m around the typical target, a fraction equal to $p_{\rm L}= 0.7$ of BSs is in LoS. Unless otherwise specified, the fraction of JCAS-capable BSs is set to $\chi_k = 0.8,\ \forall k\in\{1,\dots,K\}$. Finally, we set $A=10$ dB, $f=30$ GHz, and $\sigma_n^2=-60$ dB \cite{MUN}. It is important to note that the selection of these parameters is for the purpose of presenting the achieved performance of our proposed framework. Using different values will lead to a shifted network performance, but with the same conclusions and remarks.

\begin{figure}
	\centering\includegraphics[width=.9\linewidth]{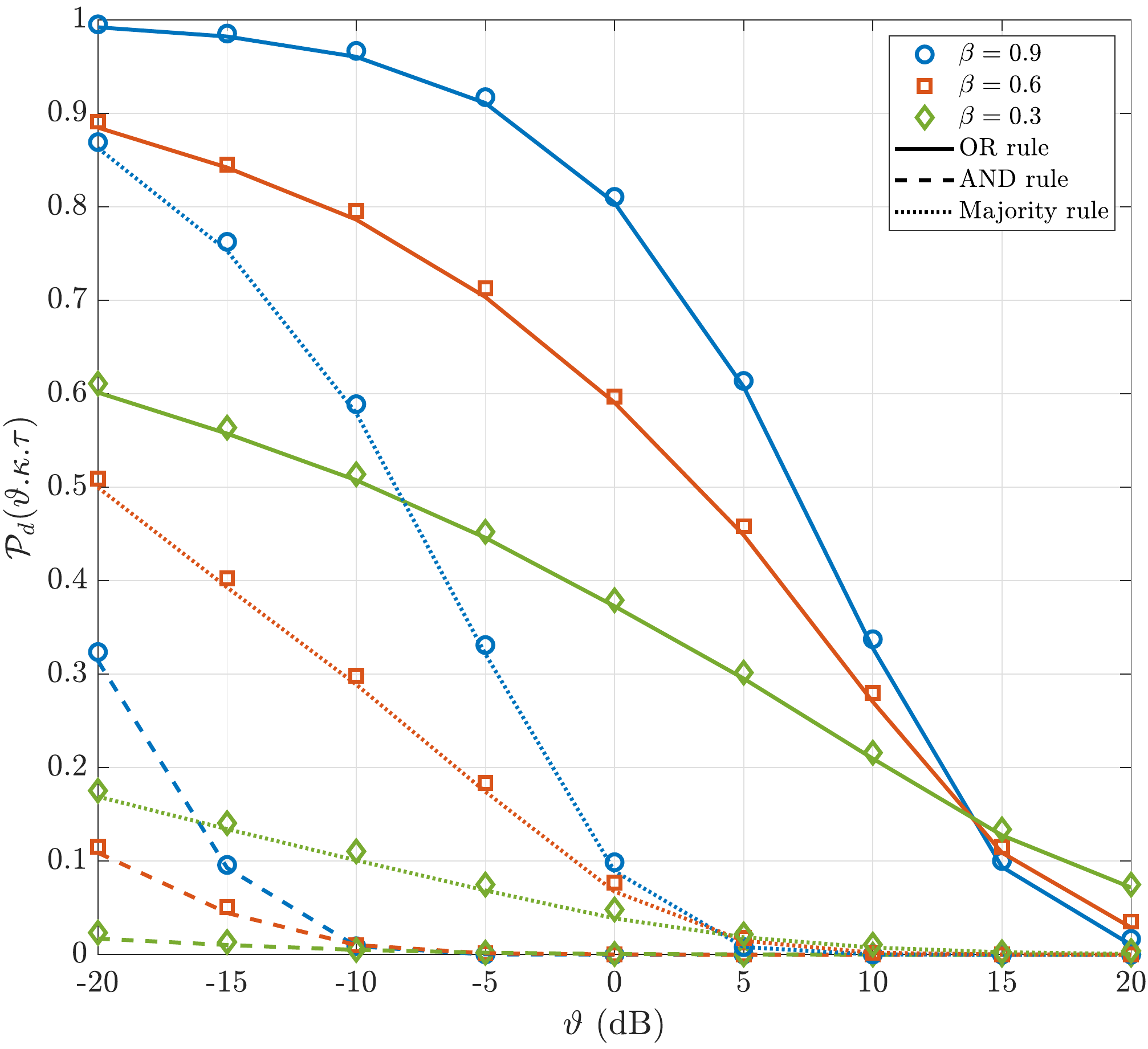}
	\caption{Detection probability versus $\vartheta$ for the considered combining rules, where $\beta=\{0.3,0.6,0.9\}$.}\label{FirstFigure}
\end{figure}

\begin{figure}
	\centering\includegraphics[width=0.92\linewidth]{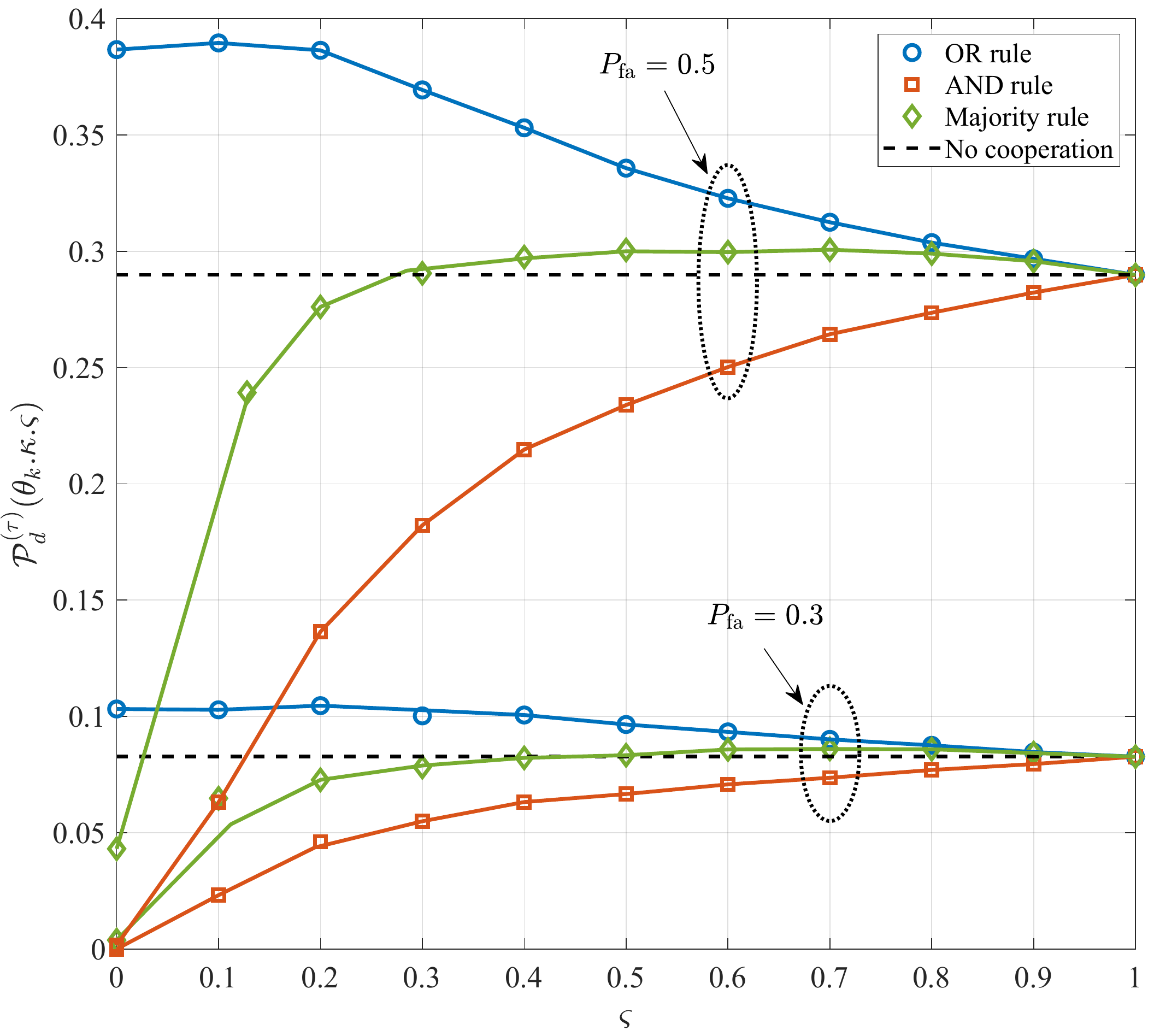}
	\caption{Detection probability versus $\varsigma$ for the considered combining rules under CFAR rate, where $P_{\rm fa} = \{0.3,0.5\}$.}\label{RespFig}
\end{figure}

Fig. \ref{FirstFigure} illustrates the detection performance achieved with our proposed technique for all three hard-decision combining rules and for different repulsion parameter values. An important observation from this figure is that the \textit{OR} combining rule achieves a significantly higher network detection performance, outperforming the other two. This observation is based on the fact that, for the scenario where the FC adopts the \textit{OR} rule, the object detection is achieved when at least one of the $K$ cooperative BSs successfully senses the existence of the particular object. On the other hand, by applying the \textit{Majority} or the \textit{AND} rule, a decreased detection performance is observed since the simultaneous successful detection of the object by the majority or all cooperative BSs, respectively, is required. Furthermore, Fig. \ref{FirstFigure} demonstrates the impact of the repulsion parameter $\beta$ on the detection performance achieved with our proposed technique. We can easily observe that, the detection performance $\mathcal{P}_d^{(\tau)}(\vartheta,\kappa,\varsigma)$ is an increasing function of $\beta$. This is expected, since by increasing the repulsion parameter, the distance between interfering BSs is increased, achieving a reduced overall interference and an enhanced detection performance. We can finally observe that the analytical results (solid, dashed and dotted lines) agree with the network performance given by the simulation results (markers).

Fig. \ref{RespFig} shows the detection performance for the adopted distance-based weighted decision rules for two constant false alarm rates $P_{\rm fa}=\{0.3,0.5\}$, where the conventional case of equal weights is depicted at $\varsigma=0$. The detection threshold $\theta_k$ is selected based on Lemma \ref{Threshold}, where $k\in\{1,\dots,K\}$. An important observation is that, the achieved detection performance of of all decision rules reduces as the tolerable false alarm rate decreases. This observation is based on the fact that, the acquisition of a lower false alarm rate requires the selection of a higher detection threshold, leading to a reduced detection performance. We can also observe that, the weighted decision rules provide a better detection performance compared to the equal-weighted detection rules for the AND and Majority rules. On the other hand, the adoption of weighted decision rules reduces the detection performance of the OR rule compared to that achieved by the equal-weighted detection rules. By comparing the performance of the proposed CoMRD technique with the conventional non-cooperative technique (dashed lines), which is obtained in the case where $\varsigma=1$, we observe that the cooperation yields relative gains in detection performance with the OR and Majority rules.

\begin{figure}
	\centering\includegraphics[width=.92\linewidth]{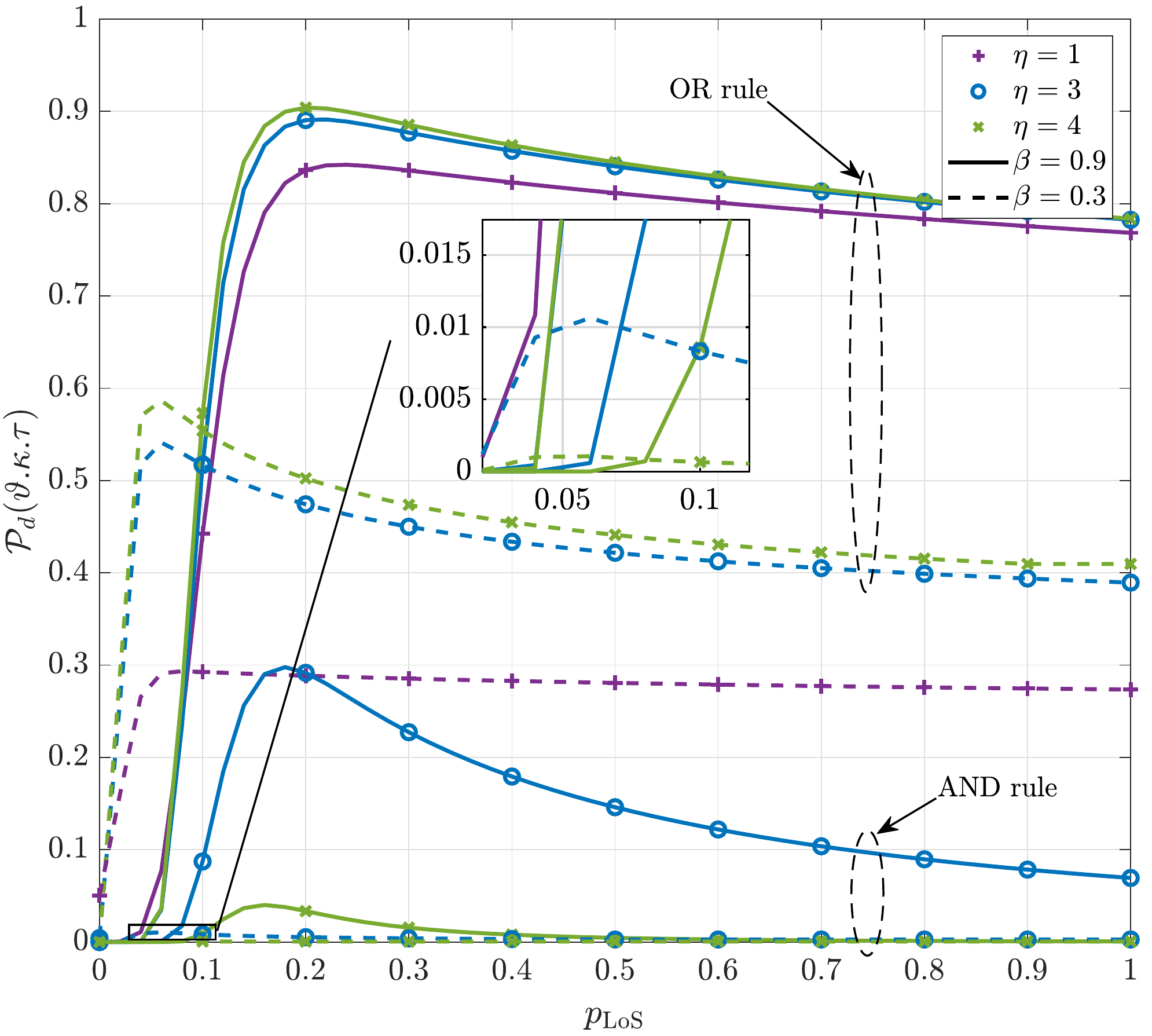}
	\caption{Detection probability versus $p_{\rm L}$ for the \textit{OR} and \textit{AND} combining rules and for different number of network tiers; $P_{\rm fa}=0.5$, $\lambda_4 = 6\ {\rm BSs/km^2}$ and $P_4 = 0$ dBm.}\label{SecondFigure}
\end{figure}
\begin{figure}
	\centering\includegraphics[width=.97\linewidth]{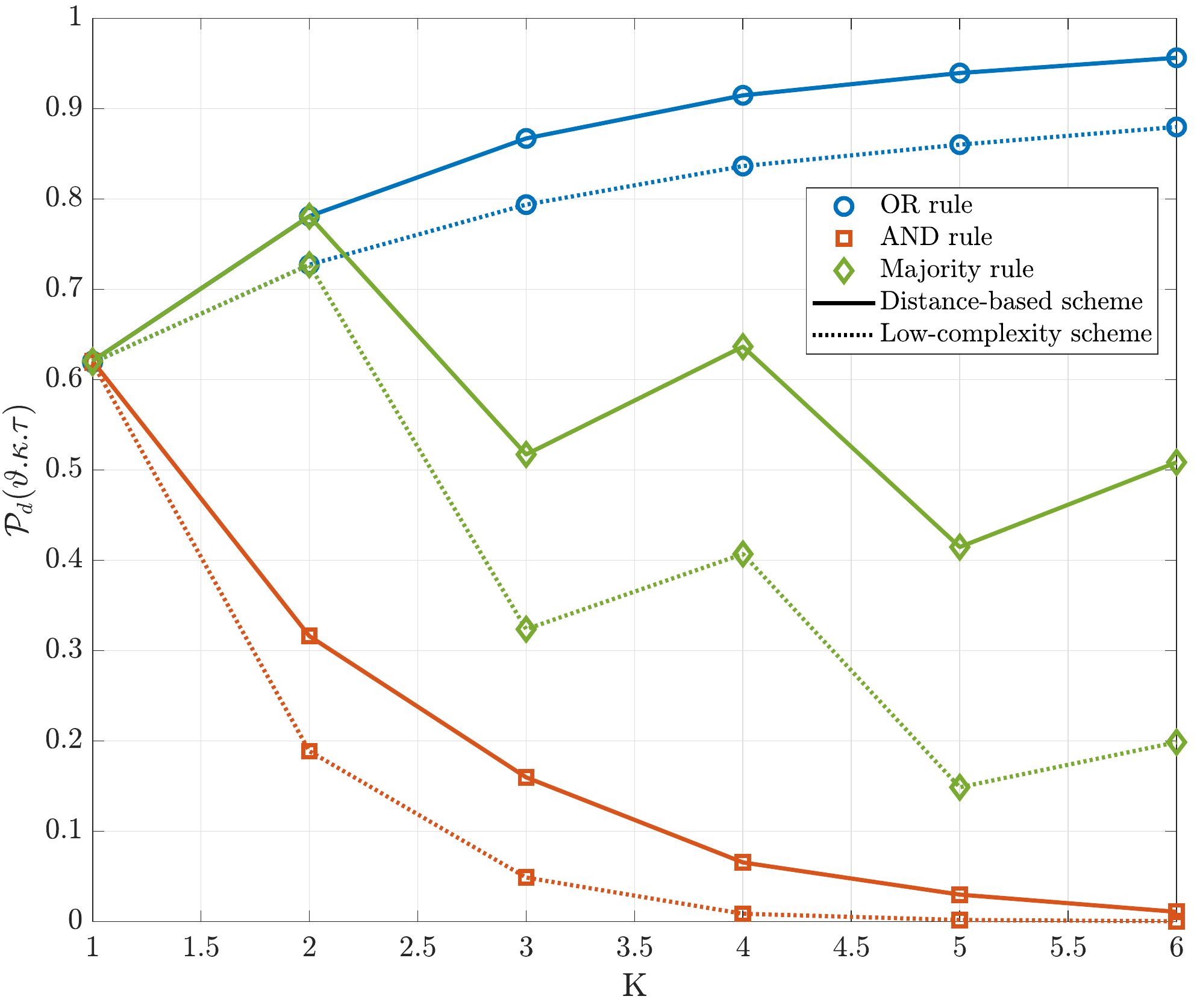}
	\caption{Detection probability versus $p_{\rm L}$ for the \textit{OR} and \textit{AND} combining rules and for different number of network tiers; $P_{\rm fa}=0.5$, $\lambda_4 = 6\ {\rm BSs/km^2}$ and $P_4 = 0$ dBm.}\label{ConventionalFigure}
\end{figure}

Fig. \ref{SecondFigure} reveals the impact of both the blockage and repulsion parameter on the detection performance achieved by applying the \textit{OR} and the \textit{AND} combining rules and for different number of network tiers. As expected, by increasing the number of network tiers, the number of cooperative BSs is also increased, and therefore the detection performance achieved by applying the \textit{OR} rule increases, in contrast to the decreasing detection performance achieved by applying the \textit{AND} rule. The latter occurs since, as the number of cooperative BSs increases, the probability of simultaneous successful detection of an object by all cooperative BSs is decreased. On the other hand, by applying the \textit{OR} rule, the increased number of cooperative BSs, increases the probability of successfully detecting the object from at least one cooperative BS. Fig. \ref{SecondFigure} also reveals the impact of blockages on the network's detection performance. As expected, at low LoS constant values, the existence of LoS BSs improves the network performance, since the cooperative BSs are able to successfully detect the object. However, by increasing the LoS constant beyond a critical point, which describes the optimal fraction of LoS BSs that provide the maximum detection performance, the network performance decreases. This observation is based on the fact that, the interference from the LoS BSs becomes significantly larger than the reflected signal from the object and thus the detection probability significantly decreases. Another important observation is that, the aforementioned critical point depends on the repulsion parameter of the network tiers. Specifically, by enhancing the repulsion behavior between the BSs, the optimal detection performance is achieved with a larger fraction of LoS BSs. Moreover, Fig. \ref{ConventionalFigure} shows the detection performance achieved with the adopted PTDB association scheme compared to the conventional DB scheme for all three hard-decision combining rules and for different number of network tiers. As it can be seen, the performance achieved with the PTDB scheme is upper bounded by the achieved performance of the DB scheme. Finally, the adoption of the PTDB scheme provides lower complexity methodology for evaluating the system performance, without being significantly deficient in accuracy.

\begin{figure}
	\centering\includegraphics[width=.9\linewidth]{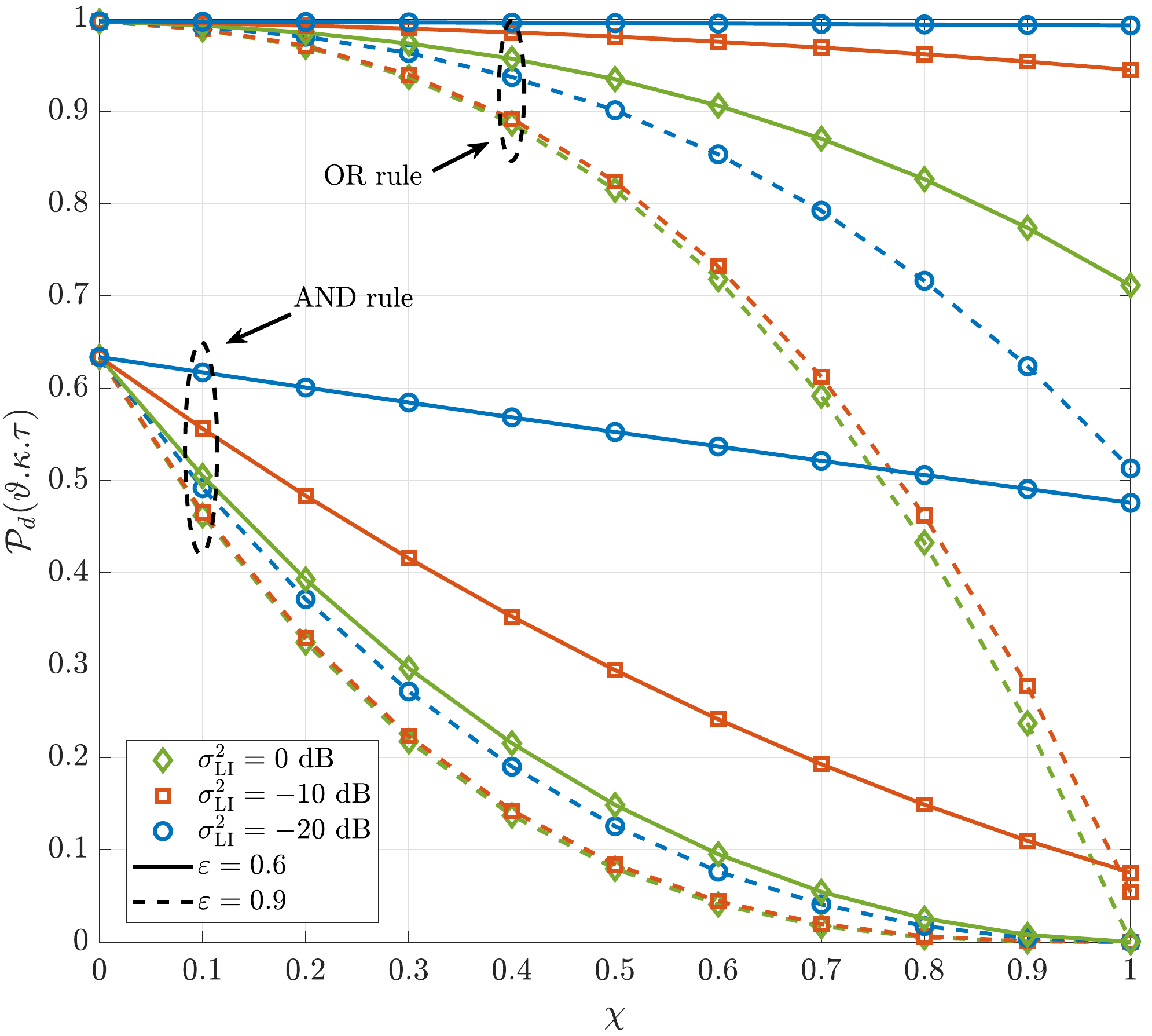}
	\caption{Detection probability versus $\chi$ for the considered combining rules, where $\sigma_{\rm SI}^2 = \{0,\ -10,\ -20\}$ dB and $\varepsilon=\{0.6\ ,0.9\}$; $\beta=0.9$, $P_{\rm fa}=0.5$.}\label{ThirdFigure}
\end{figure}

Fig. \ref{ThirdFigure} shows the effect of the residual SI on the network's detection performance for the \textit{OR} and the \textit{AND} combining rules. We can easily observe that by increasing the ability of the nodes to cancel the observed SI i.e., $\sigma_{\rm SI}^2\rightarrow -\infty$, the detection performance achieved with our proposed technique is increased for all combining rules. This observation was expected, since by decreasing the residual SI at the BSs, the aggregate received interference is decreased, and therefore an increased SINR is observed. An important observation from this figure is that the increased fraction of JCAS-enabled BSs, negatively affects the detection performance. This again is expected, since by increasing the number of BSs that simultaneously perform both functions i.e., DL transmission and detection, the overall interference increases due to the existence of SI, compromising the observed SINR. We can easily observe from the figure that, the power control can prevent the significant degradation of the network's detection performance by controlling the power control parameter $\varepsilon$. Specifically, by decreasing the power control parameter i.e., $\varepsilon\rightarrow 0$, the BSs' transmit power for the DL communication with their associated users is reduced, resulting in a decreased SI and therefore in an enhanced SINR. It is important to mention here that both the fraction of JCAS-enabled BSs and the Nakagami-$\mu$ fading parameters, have similar effect on the detection performance achieved with the \textit{Majority} rule, thus the corresponding curves are omitted.
\begin{figure}
	\centering\includegraphics[width=.9\linewidth]{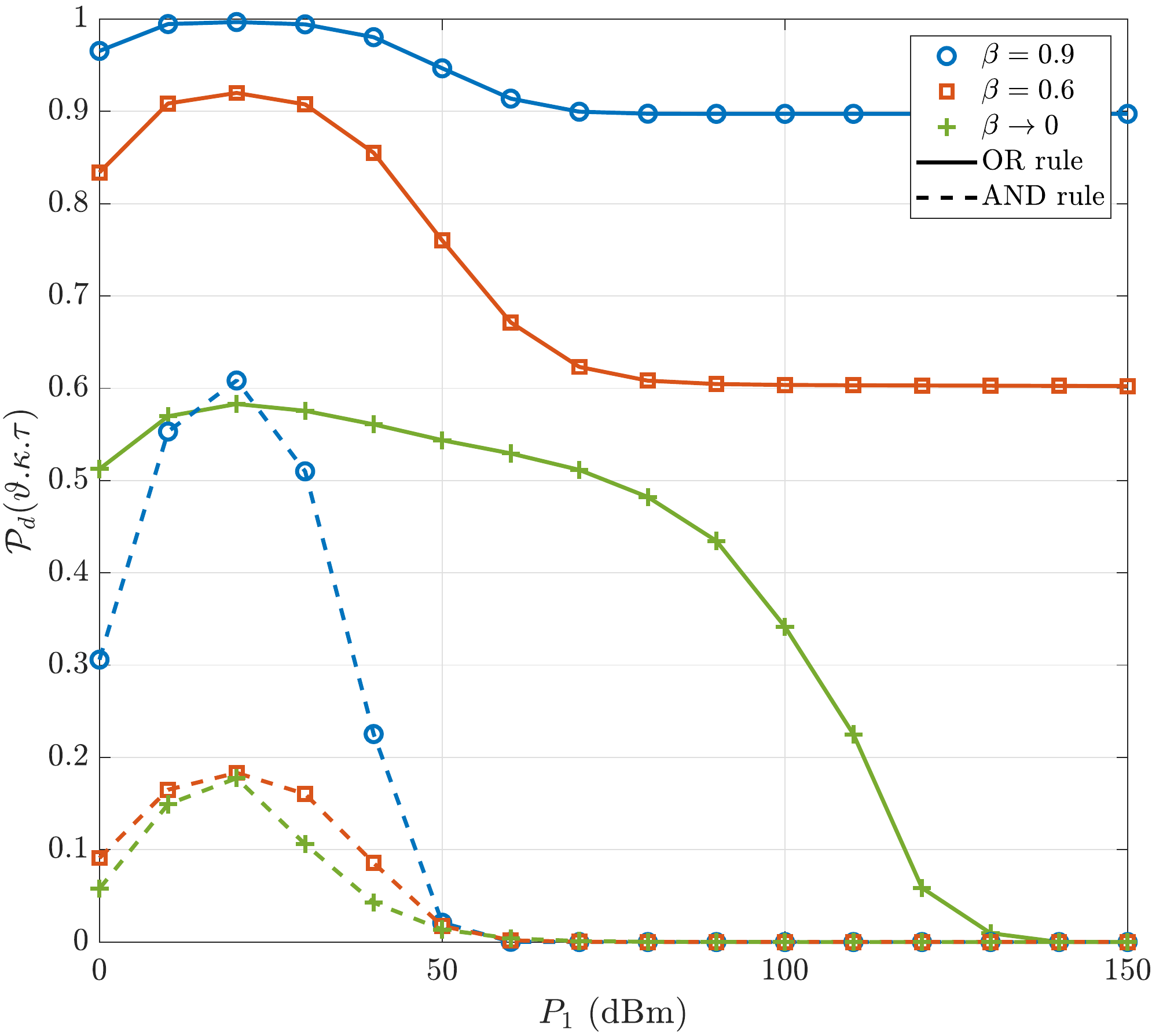}
	\caption{Detection probability versus $P_1$ for the OR and the AND combining rules, where $\beta = \{0,0.6,0.9\}$; $P_2 = \frac{P_1}{2}$ dB, $P_3 = \frac{P_1}{3}$ dB and $\chi=0.3$.}\label{ForthFigure}
\end{figure}

Fig. \ref{ForthFigure} shows the detection performance of the OR and the AND combining rules in terms of the BSs' transmit power $P_1$, where $P_2 = \frac{P_1}{2}$ and $P_3 = \frac{P_1}{3}$. We can easily observe from the figure that, by increasing the transmit power of the BSs can significantly enhance the detection probability of its associated object. This observation is expected since, a larger transmission power by the cooperative BSs corresponds to receiving a stronger reflected signal power, resulting in an enhanced detection performance. Nevertheless, beyond a critical transmit power level, the increase of the transmission power negatively affects the detection performance. This is again expected, since the residual SI at a receiver becomes larger, and therefore the observed SINR is decreased. Moreover, it can be seen that the detection performance converges to a constant floor in all cases for high transmission powers. This is due to the fact that as the transmission power of the network’s BSs increases, the noise in the network becomes negligible. Fig. \ref{ForthFigure} also illustrates the achieved detection performance of the considered combining rules, for the scenario where the spatial distribution of the BSs is based on a PPP. It is clear from the figure that, the spatially distributed BSs based on a repulsive point process can attain better detection performance compared to the randomly deployed BSs. This is due to the fact that, the distances between a cooperative BS and its interfering BSs in a $\beta$-GPP-based deployment are greater than the corresponding distances in a PPP-based deployment, resulting in a reduced overall interference.
\begin{figure}
	\centering\includegraphics[width=.9\linewidth]{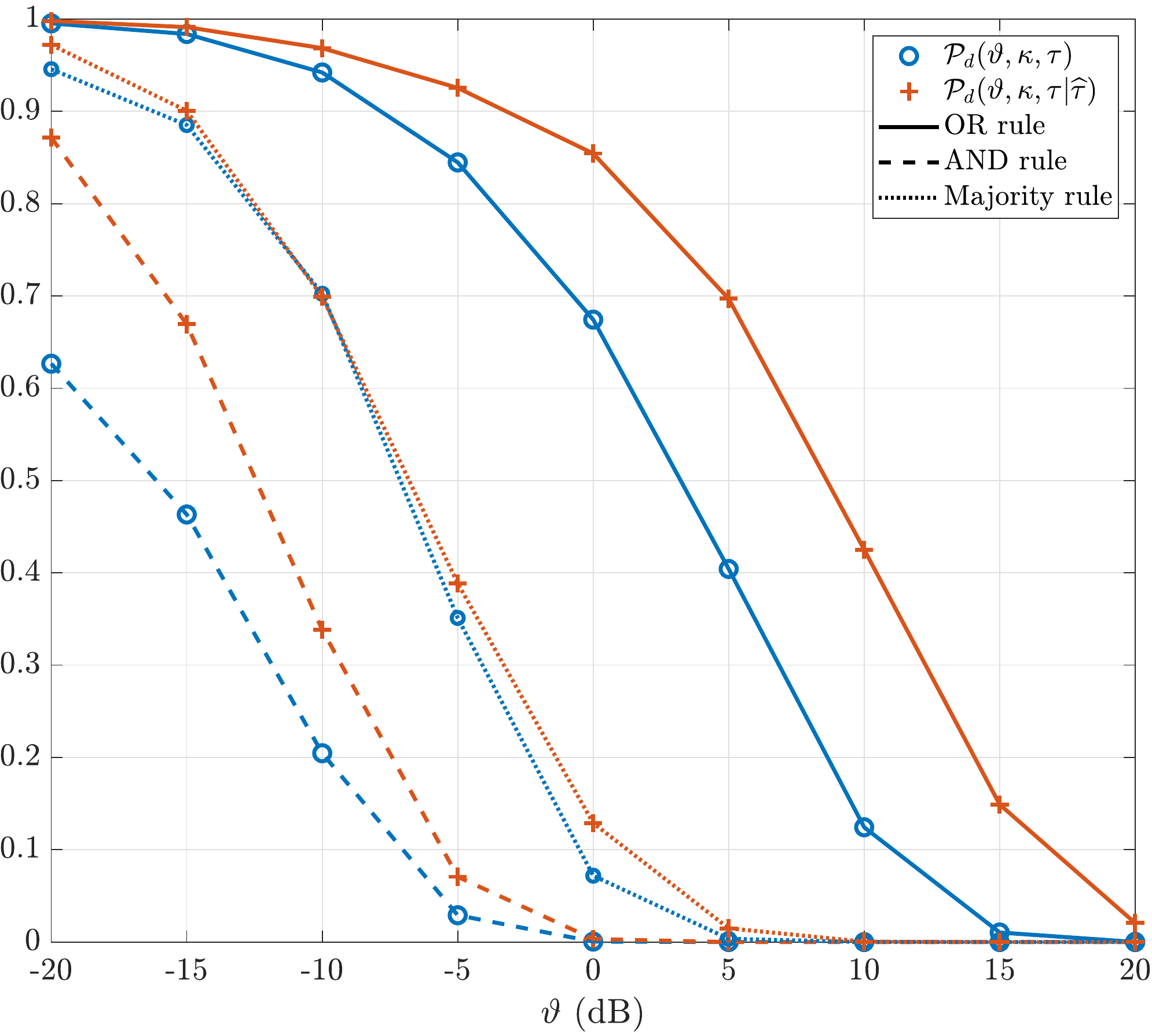}
	\caption{Detection probability and conditional detection probability versus $\vartheta$ for the considered combining rules; $\beta=0.9$, $\chi_k=0.3$.}\label{FifthFigure}
\end{figure}

Finally, Fig. \ref{FifthFigure} highlights the effect of the interference temporal correlation on the network's detection performance between different time slots. Specifically, we plot the conditional and the unconditional detection probabilities for the three considered combining rules. We can easily observe that the conditional detection probability overcomes the unconditional detection probability. This was expected since, a fraction of interfering BSs at the time slot $\tau$, in which a static object is successfully detected, also causes interference in future time slots $\widehat{\tau}$. Therefore, as the static object is successfully detected at the time slot $\tau$, the probability of successful detecting the same static object at the time slot $\widehat{\tau}$ is increased. Note that, the aforementioned observation holds for all the considered combining rules and for all the detection thresholds considered.

\section{Conclusion}\label{Section7}
In this paper, we proposed an analytical framework based on stochastic geometry and studied the detection performance of FD-JCAS systems in the context of heterogeneous mmWave cellular networks. In particular, we considered the scenario where all BSs exhibit detection capabilities, a fraction of which also exhibit DL communication capabilities by exploiting the concept of FD radio. The developed framework takes into account the spatial correlation between the BSs and the temporal correlation of the interference, by modeling the BSs' locations as a $\beta$-GPP and by analyzing the joint detection probability at two different time slots, respectively. Aiming at enhancing the network's detection performance, a novel cooperative detection technique is proposed, and the achieved network performance is evaluated in the context of three hard-decision combining rules, namely the \textit{OR}, the \textit{Majority} and the \textit{AND} rule. Using stochastic geometry tools, the network's detection performance in the context of our proposed technique was derived in analytical expressions and the impact of blockage characteristics, residual SI, repulsion parameter and fraction of FD-JCAS BSs have been evaluated. Our study reveals that the proposed technique outperforms the conventional non-cooperative detection technique. Finally, we have shown that the repulsive behavior of the BSs can provide significantly larger detection performance to the network.

\begin{IEEEbiography}[{\includegraphics[width=1in,height=1.25in,clip,keepaspectratio]{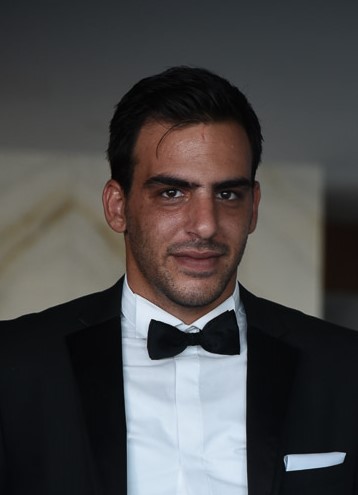}}]{Christodoulos Skouroumounis} (S'15–M'20) received the diploma in Computer Engineer from the Electrical and Computer Engineer Department of National Technical University of Athens (NTUA), Greece, in 2014, and a Ph.D. in Computer Engineer from the University of Cyprus, Cyprus in 2019. He is currently a Post-Doctoral Research Fellow with the Department of Electrical Engineering, Computer Engineering and Informatics, Cyprus University of Technology. His current research interests include full-duplex radio, next-generation communication systems and cooperative networks.
\end{IEEEbiography}

\begin{IEEEbiography}[{\includegraphics[width=1in,height=1.25in,clip,keepaspectratio]{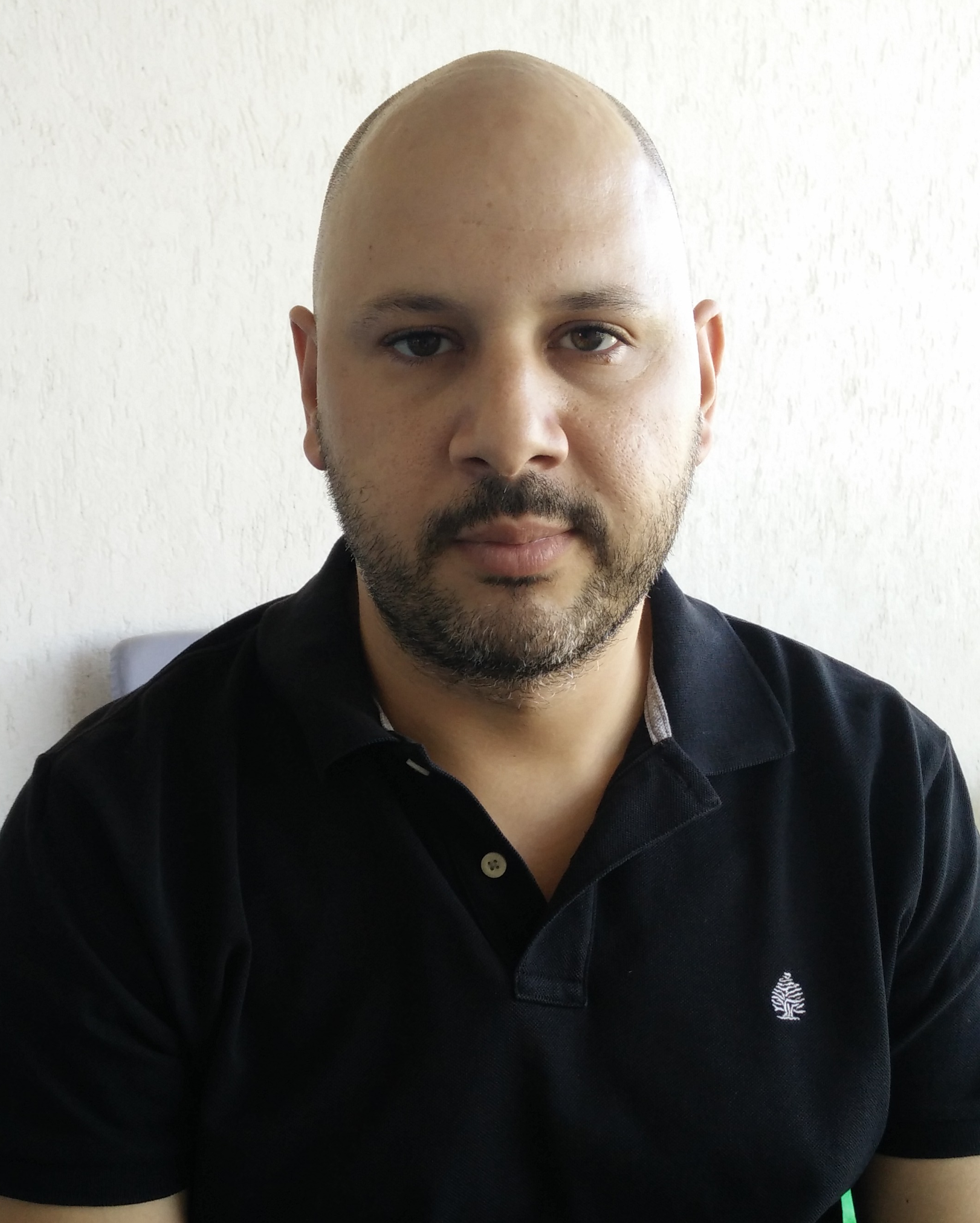}}]{Constantinos Psomas} (M'15–SM'19) received the BSc in Computer Science and Mathematics from Royal Holloway, University of London in 2007 and the MSc in Applicable Mathematics from London School of Economics the following year. In 2011, he received a PhD in Mathematics from The Open University, UK. He is currently a Research Fellow at the Department of Electrical and Computer Engineering of the University of Cyprus. From 2011 to 2014, he was as a Postdoctoral Research Fellow at the Department of Electrical Engineering, Computer Engineering and Informatics of the Cyprus University of Technology. Dr. Psomas serves as an Associate Editor for the IEEE Wireless Communications Letters. He received the Exemplary Reviewer Award by the IEEE Wireless Communications Letters in 2015 and 2018. His current research interests include wireless powered communications, cooperative networks and full-duplex communications.
\end{IEEEbiography}

\begin{IEEEbiography}[{\includegraphics[width=1in,height=1.25in,clip,keepaspectratio]{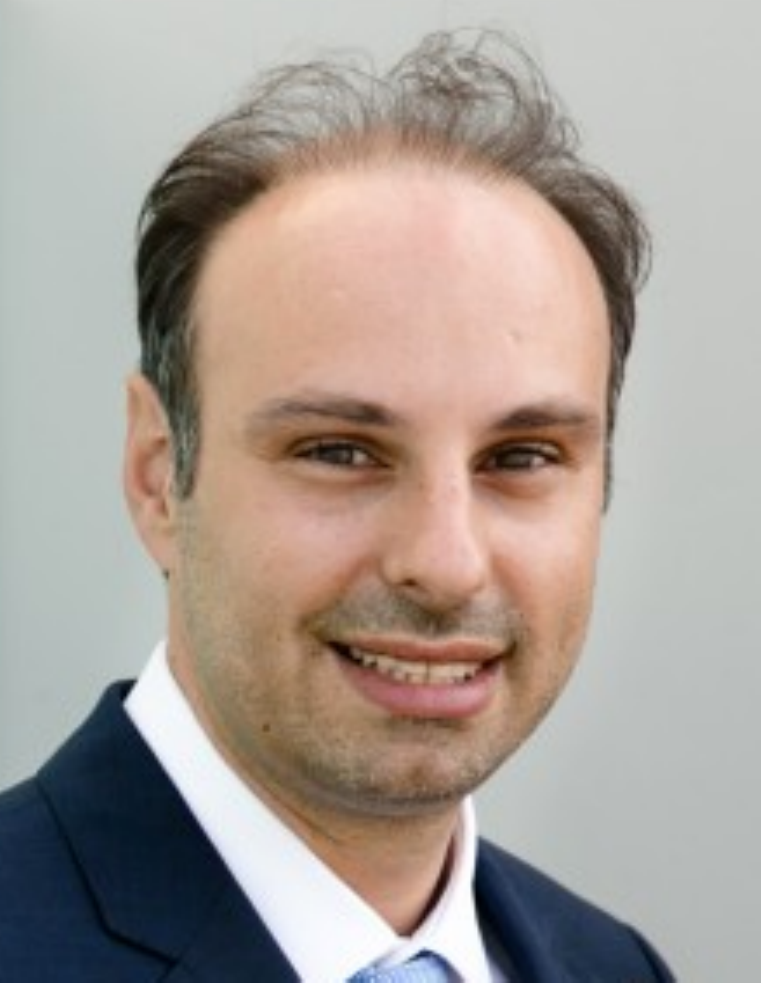}}]{Ioannis Krikidis}(S'03–M'07–SM'12–F'19) received the Diploma degree in computer engineering from the Computer Engineering and Informatics Department (CEID), University of Patras, Greece, in 2000, and the M.Sc. and Ph.D. degrees from \'{E}cole Nationale Sup\'{e}rieure des T\'{e}l\'{e}communications (ENST), Paris, France, in 2001 and 2005, respectively, all in electrical engineering. 
	
From 2006 to 2007, he worked as a PostDoctoral Researcher with ENST, Paris. From 2007 to 2010, he was a Research Fellow of the School of Engineering and Electronics, The University of Edinburgh, Edinburgh, U.K. He also has held research positions at the Department of Electrical Engineering, University of Notre Dame; the Department of Electrical and Computer Engineering, University of Maryland; the Interdisciplinary Centre for Security, Reliability and Trust, University of Luxembourg; and the Department of Electrical and Electronic Engineering, Niigata University, Japan. He is currently an Associate Professor at the Department of Electrical and Computer Engineering, University of Cyprus, Nicosia, Cyprus. His current research interests include wireless communications, cooperative networks, 5G communication systems, wireless powered communications, and secrecy communications. He serves as an Associate Editor for IEEE TRANSACTIONS ON WIRELESS COMMUNICATIONS, IEEE TRANSACTIONS ON GREEN COMMUNICATIONS AND NETWORKING, and IEEE WIRELESS COMMUNICATIONS LETTERS. He was a recipient of the Young Researcher Award from the Research Promotion Foundation, Cyprus, in 2013, and a recipient of IEEE ComSoc Best Young Professional Award in Academia, 2016. He has been recognized by the Web of Science as a Highly Cited Researcher for 2017, 2018, and 2019. He has received the prestigious ERC Consolidator Grant.
\end{IEEEbiography}
\end{document}